\documentclass[letterpaper, 10pt, conference]{ieeeconf}      % Use this line for a4
\IEEEoverridecommandlockouts                              % This command is only
\usepackage{amsmath,mathrsfs,amsfonts,amssymb,graphicx,epsfig}
 \usepackage{amsthm}
\usepackage{subcaption}
\usepackage{color,multirow,rotating}
\usepackage{algorithm,algpseudocode,algorithmicx}
\usepackage{cite}
\usepackage{framed}
\usepackage{bm}

\newtheorem{theorem}{Theorem}

\usepackage{tikz}
\usetikzlibrary{calc,trees,positioning,arrows,chains,shapes.geometric,decorations.pathreplacing,decorations.pathmorphing,shapes, matrix,shapes.symbols}

\DeclareMathOperator*{\atantwo}{atan2}
\DeclareMathOperator*{\diag}{diag}
\DeclareMathOperator*{\blkdiag}{blkdiag}

\newcommand{\bx}{{\bm{x}}}

\newcommand{\differential}{{\rm{d}}}
\renewcommand{\det}{{\mathrm{det}}}

\allowdisplaybreaks

\title{\LARGE\textbf{Prediction and Optimal Feedback Steering of Probability Density\\Functions for Safe Automated Driving 
}
}
\author{Shadi Haddad, Kenneth F. Caluya, Abhishek Halder, and Baljeet Singh% <-this % stops a space
\thanks{Shadi Haddad, Kenneth F. Caluya, and Abhishek Halder are with the Department of Applied Mathematics, University of California, Santa Cruz, CA 95064, USA,
        {\tt\small{\{shhaddad,kcaluya,ahalder\}@ucsc.edu}}. Baljeet Singh is with the Ford Greenfield Labs, Palo Alto, CA 94304, USA, {\tt\small{BSING124@ford.com}}.%
}}

\begin{document}

\maketitle
\thispagestyle{empty}
\pagestyle{empty}

\def\spacingset#1{\def\baselinestretch{#1}\small\normalsize}
\spacingset{1}

\begin{abstract}
We propose a stochastic prediction-control framework to promote safety in automated driving by directly controlling the joint state probability density functions (PDFs) subject to the vehicle dynamics via trajectory-level state feedback. To illustrate the main ideas, we focus on a multi-lane highway driving scenario although the proposed framework can be adapted to other contexts. The computational pipeline consists of a PDF prediction layer, followed by a PDF control layer. The prediction layer performs moving horizon nonparametric forecasts for the ego and the non-ego vehicles' stochastic states, and thereby derives safe target PDF for the ego. The latter is based on the forecasted collision probabilities, and promotes the probabilistic safety for the ego. The PDF control layer designs a feedback that optimally steers the joint state PDF subject to the controlled ego dynamics while satisfying the endpoint PDF constraints. Our computation for the PDF prediction layer leverages the structure of the controlled Liouville PDE to evolve the joint PDF values, as opposed to empirically approximating the PDFs. Our computation for the PDF control layer leverages the differential flatness structure in vehicle dynamics. We harness recent theoretical and algorithmic advances in optimal mass transport, and the Schr\"{o}dinger bridge. The numerical simulations illustrate the efficacy of the proposed framework. 
\end{abstract}

\section{Introduction}\label{SecIntro}
We propose a two layer framework for the prediction and feedback control of joint state probability density functions to promote stochastic safety in multi-lane highway driving scenarios such as Fig. \ref{fig:Schematic}. While recent works \cite{carvalho2015automated,funfgeld2017stochastic} have advocated the use of stochastic forecasts for safety considerations in automated driving, typical Monte Carlo-based predictive algorithms incur high computational cost due to the curse of dimensionality. Concomitantly, the role of feedback control in automated driving has been limited to mitigating, rather than active steering of uncertainties.

Building on \cite{haddad2020density}, our \emph{first contribution} is to show that a direct solution of the characteristic ODEs associated with certain Liouville PDEs, allow propagation of the joint state PDFs of the ego and the non-ego vehicles in the form of probability-weighted scattered point clouds. Such a computation does not require approximating the nonlinearities of the vehicle dynamics, or the time varying statistics.

Our \emph{second contribution} is to directly design feedback controllers for regulating the ego vehicle's stochastic states from a given initial joint PDF to a desired terminal joint PDF in finite horizon while minimizing the control effort. Based on the collision probability computation, we infer whether a lane change could be safer for the ego, and if so, we derive the desired terminal PDF as the ``safest" Wasserstein barycenter of the two consecutive non-ego vehicles in another lane. We exploit the differential flatness structure of the underlying dynamics, as well as recent theoretical advances \cite{chen2016optimal,chen2016entropic,caluya2020finite} in density control. To speed up the controller synthesis computation, we derive closed form formula for the inverse and determinant of the finite horizon controllability Gramian associated with the Brunovsky normal form--these results should be of independent interest.

\begin{figure}
\centering
\begin{subfigure}[b]{\linewidth}
   \includegraphics[width=\linewidth]{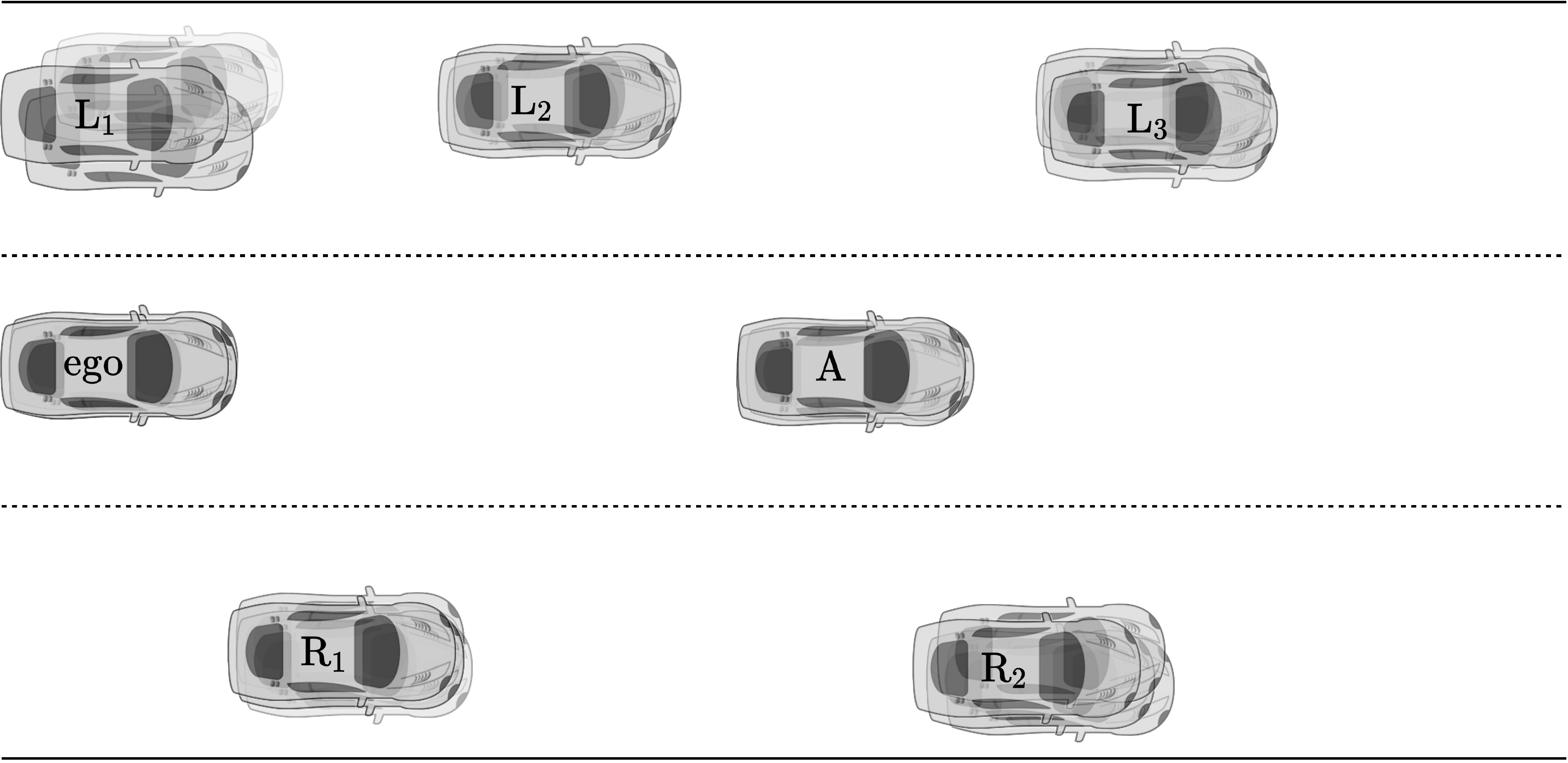}
   \caption{{\small{The ego vehicle's estimate at time $t=t_{0}$.}}}
   \label{fig:InitialTimeSchematic} 
\end{subfigure}
\par\bigskip
\begin{subfigure}[b]{\linewidth}
   \includegraphics[width=\linewidth]{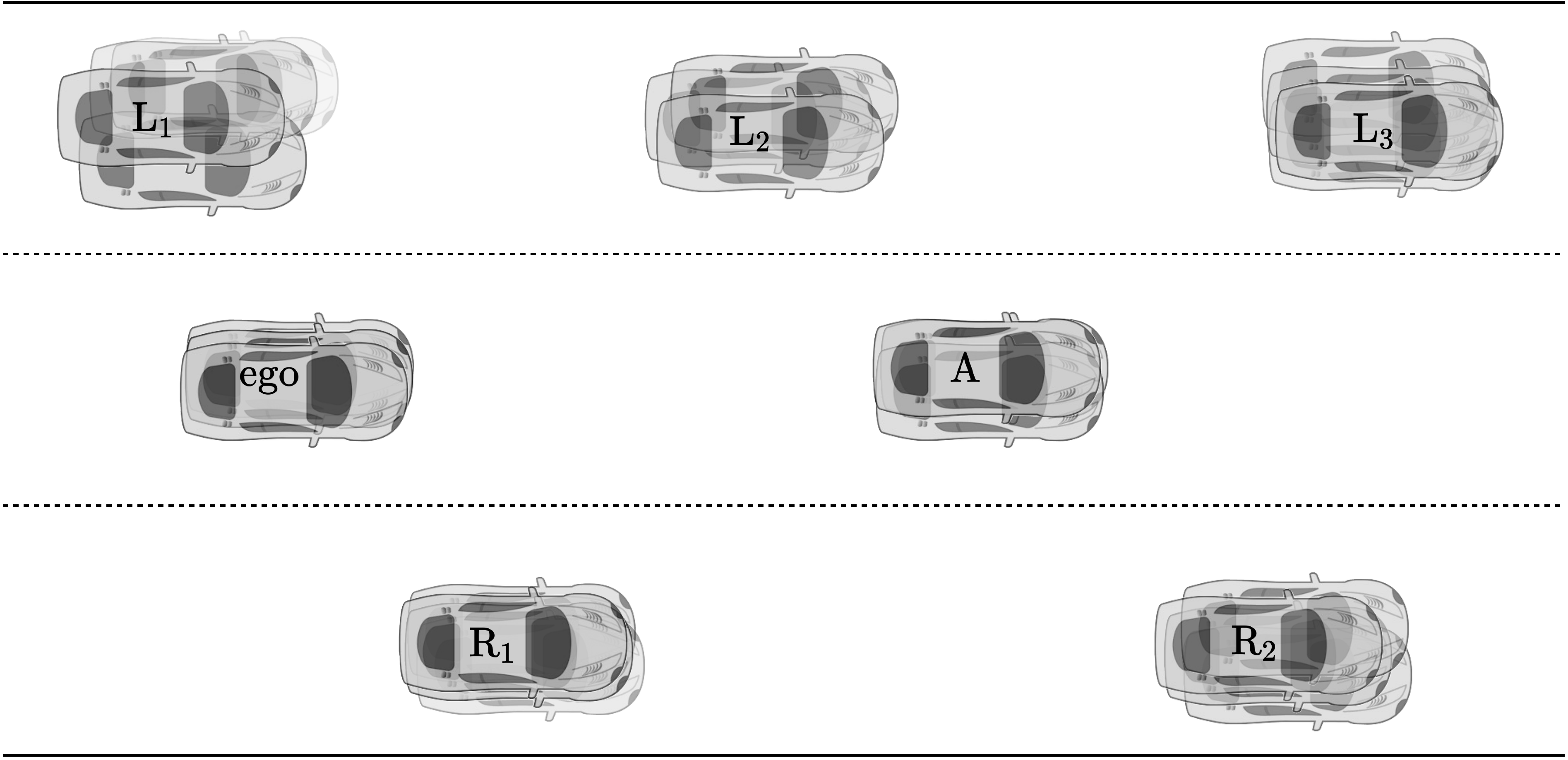}
   \caption{{\small{The ego vehicle's prediction for $t=t_{0}+T$ made at time $t=t_{0}$.}}}
   \label{fig:FinalTimeSchematic}
\end{subfigure}
\caption{{\small{A schematic of the stochastic states in multi-lane unidirectional highway driving scenario viewed from the ego vehicle's perspective. The dashed lines denote the lane boundaries. The three cars in the left of the ego vehicle's lane are labeled as L\textsubscript{1}, L\textsubscript{2}, L\textsubscript{3}. Likewise, the two cars in the right of the ego vehicle's lane are labeled as R\textsubscript{1}, R\textsubscript{2}. The car ahead of the ego is labeled as A. (a) At $t=t_{0}$, the ego vehicle's \emph{estimates} of the stochastic states or beliefs of all cars (including itself) in its neighborhood. (b) At $t=t_{0}$, the ego vehicle's \emph{predictions} of the stochastic states or beliefs at $t_{0}+T$ of all cars (including itself) in its neighborhood.}}}
\vspace*{-0.25in}
\label{fig:Schematic}
\end{figure}

\begin{figure*}[t]
\centering
   \includegraphics[width=\linewidth]{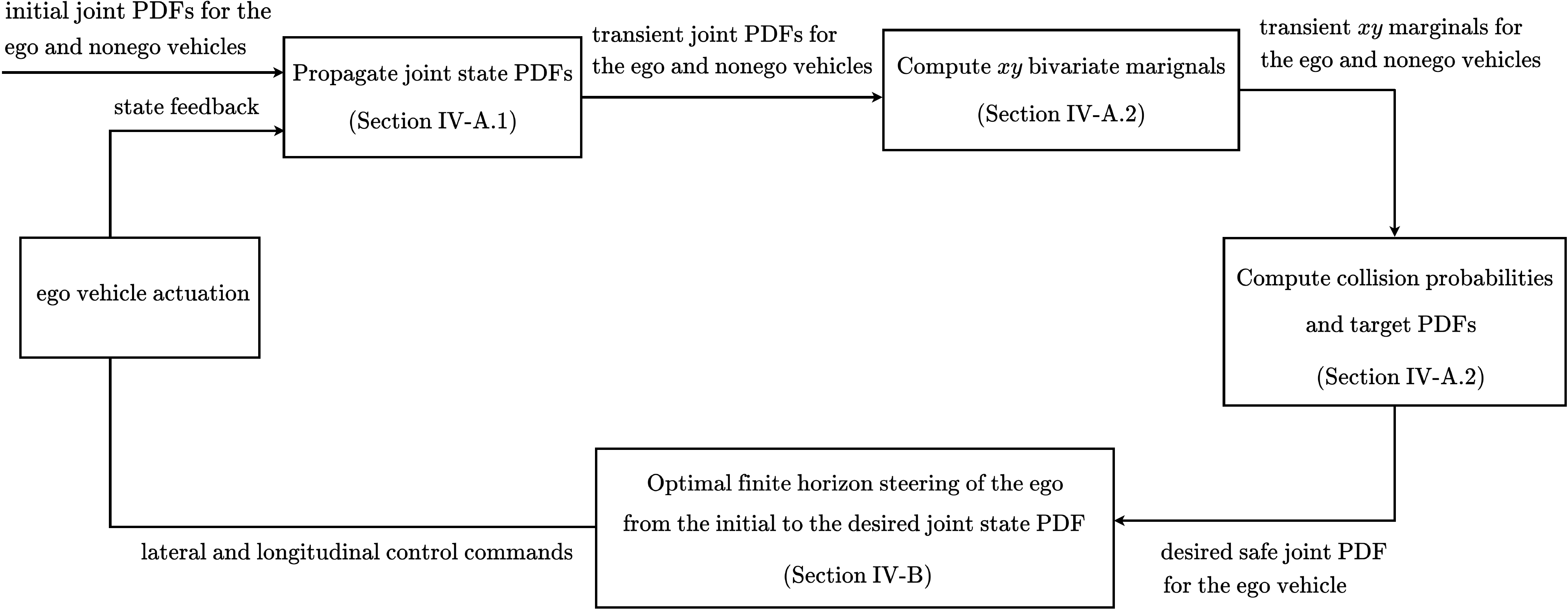}
   \caption{{\small{A block diagram of the proposed PDF prediction-control framework for multi-lane highway driving scenarios such as in Fig. \ref{fig:Schematic}.}}}
   \label{fig:Architecture} 
   \vspace*{-0.15in}
\end{figure*}

\subsubsection*{Notations} Symbols for matrices and vectors are set to be boldfaced capital and small, respectively. Non-boldfaced variables are scalars. We use $\left[\:\cdot\:\right]_{ij}$ to denote the $(i,j)$th element of a matrix. For a block diagonal matrix, we use the symbol $\blkdiag(\cdot)$ whose arguments are the diagonal blocks. For a given vector $\bm{v}$, the symbol $\diag(\bm{v})$ denotes a diagonal matrix whose diagonal comprises of the elements of $\bm{v}$. The notation $\bm{e}_{i}^{d}$ stands for the $i$th standard basis vector in $\mathbb{R}^{d}$ while $\bm{0}$ denotes a column vector of zeros in appropriate dimension. The symbols $\nabla$, ${\rm{Hess}}(\cdot)$, $\det(\cdot)$, $!$, $\oslash$ stand for the Euclidean gradient, Hessian, determinant, factorial, and Hadamard (entry-wise) division, respectively. We use $\Gamma(\cdot)$ to denote the Gamma function. For any positive integer $n$, we have $\Gamma(n) = (n-1)!$. Furthermore, the symbol $\sharp$ denotes the pushforward of PDF, or probability measure in general. The expectation operator w.r.t. the joint PDF $\rho(\bm{x})$ is denoted as $\mathbb{E}_{\rho}\left[\cdot\right] := \int \left(\cdot\right)\rho\:\differential\bm{x}$. For matrix arguments, the symbol $\langle\cdot,\cdot\rangle$ denotes the Frobenius inner product.

The remaining of this paper is structured as follows. Section \ref{SecArchitecture} provides an overview of the proposed framework. We delineate the vehicle model in Section \ref{SecModels}. The algorithmic details of the proposed two layer prediction-control framework are explained in Section \ref{SecFramework}. The numerical simulations are given in Section \ref{SecNumSim}. Section \ref{SecConclusions} contains the concluding remarks. 

%%%%%%%%%%%%%%%%%%%%%%%%%%%%%%%%%%%%%%%%%%%%%%%%%%%%%%%%%%%%%%%%%%%%%%%%%%%%%%%%%%%%%%%%%%%%%%%%%%%%%%%%%%%%%%%%%%%%%

\section{An Overview of the Proposed Framework}\label{SecArchitecture}
This paper is concerned with the following problem: given the dynamic stochastic uncertainties in the vehicles' states in multi-lane highway driving scenarios such as in Fig. \ref{fig:Schematic}, how should an individual vehicle make safe online decisions in real time? In particular, \emph{what} should the vehicle compute, \emph{when}, and \emph{how}? To address the same, we propose a computational framework for the prediction and control of joint PDFs. 

The block diagram in Fig. \ref{fig:Architecture} shows a high level overview of the proposed framework. Specifically, we propose (Sec. \ref{SubsecPredictionLayer}) to predict the joint state PDFs subject to the closed-loop vehicle dynamics by solving the associated Liouville PDE via method-of-characteristics. We suppose that the ego vehicle can estimate its own, as well as other non-ego vehicles' state PDFs at the beginning of the prediction horizon, using standard sensor fusion algorithms such as the particle filter. Since the stochastic prediction is performed in a moving horizon manner, these ``initial" joint PDFs are updated at the beginning of each prediction horizon. The predicted joint PDFs allow computing the transient marginals, which in turn are leveraged (Sec. \ref{subsubsecCollisionProbWassBary}) to compute the time-varying collision probabilities, followed up by computing a desired safe PDF for the ego vehicle at the end of the prediction horizon. We then synthesize a feedback controller (Sec. \ref{SubsecControlLayer}) that steers the ego vehicle's stochastic state from the initial joint to the computed desired joint PDF over the fixed time horizon while minimizing the control effort and respecting the vehicle dynamics. 

We remark here that the optimal feedback synthesis for density steering is a nonstandard stochastic optimal control problem since it calls for solving a two point boundary value problem on the space of probability measures. Our solution approach--stochastic regularization (Sec. \ref{subsubsecStocRegularization}) followed by certain fixed point recursion (Sec. \ref{subsubsecFixedPointRecursion})--harnesses recent theoretical progress \cite{caluya2019wasserstein} for this class of problems as well as the differential flatness structure of the vehicle dynamics. 

Deploying the framework in Fig. \ref{fig:Architecture} requires scalable algorithms for each of the blocks shown. As alluded in Section \ref{SecIntro}, existing literature advocating the use of stochastic forecasts in automated driving rely on Monte Carlo, which are expensive for onboard computation considering the $O(1)$ seconds physical dynamics timescale. Also, increasing the accuracy of the Monte Carlo requires finer discretization of the state space, thereby the computational time scales exponentially with the state dimension. In contrast, the proposed algorithms in Sec. \ref{SecFramework} are gridless in the sense they work with probability weighted point clouds. The probability weights are explicitly computed along the trajectories. For computational speed and accuracy comparisons in prediction, we refer the readers to \cite[Sec. IV-C]{haddad2020density}. For the feedback controller in Sec. \ref{SubsecControlLayer}, notice that the sensing and actuation remain in the signal level, as in standard model predictive control (MPC). However, unlike the standard MPC, our feedback controller guarantees exact stochastic steering in distribution sense. 

Directly formulating the prediction and control problems at the PDF level allows us to manage the nonparametric stochastic uncertainties in a quantitative manner, as opposed to the standard practice of applying the MPC for a nominal system and then verifying the statistical performance in numerical simulation.

%%%%%%%%%%%%%%%%%%%%%%%%%%%%%%%%%%%%%%%%%%%%%%%%%%%%%%%%%%%%%%%%%%%%%%%%%%%%%%%%%%%%%%%%%%%%%%%%%%%%%%%%%%%%%%%%%%%%%

\section{Bicycle Model and Differential Flatness}\label{SecModels}
To describe an individual car's trajectory-level, i.e., microscopic motion, we consider the kinematic bicycle model:
\begin{align}
\dot{x} = v\cos\theta,\quad \dot{y} = v\sin\theta,\quad\dot{\theta} = \frac{v}{\ell}\tan\phi,\quad\dot{v} = a,
\label{KinemaicBicycleModelSimplified}	
\end{align}
with state vector $\bx:=(x,y,\theta,v)^{\top}$ comprising of the longitudinal and lateral position coordinates of the center of the rear axle $(x,y)$, the heading angle $\theta$, and the speed $v$. The control vector $\bm{u}:=(a,\phi)^{\top}$ comprises of the acceleration $a$, and the steering wheel angle $\phi$. The parameter $\ell$ denotes the distance between the front and rear axles. The model (\ref{KinemaicBicycleModelSimplified}) assumes small sideslip angle.

It is well-known that the model (\ref{KinemaicBicycleModelSimplified}) is differentially flat \cite{fliess1995flatness,fliess1999lie,van1998differential,murray1995differential}. Specifically, taking the $2\times 1$ vector $\bm{\eta}\equiv(\eta_{1},\eta_{2})^{\top}:=(x,y)^{\top}$ as the flat output, all states and controls can be written in terms of $\bm{\eta}$ and its time derivatives:
\begin{equation}
\begin{aligned}
&x = \eta_{1}, \quad y = \eta_2, \quad \theta = \atantwo(\dot{\eta}_{2},\dot{\eta}_{1}),\\
&v = \!\left(\dot{\eta}_{1}^{2} + \dot{\eta}_{2}^{2}\right)^{\!\frac{1}{2}}, \quad a = \left(\dot{\eta}_{1}\ddot{\eta}_{1} + \dot{\eta}_{2}\ddot{\eta}_{2}\right)\!\!/\!\!\left(\dot{\eta}_{1}^{2} + \dot{\eta}_{2}^{2}\right)^{\!\frac{1}{2}\!},\\
&\phi = \arctan\!\left(\ell(\dot{\eta}_1\ddot{\eta}_2 - \dot{\eta}_2\ddot{\eta}_1)/\!\!\left(\dot{\eta}_{1}^{2} + \dot{\eta}_{2}^{2}\right)^{\!\frac{3}{2}}\right). 
\end{aligned}
\label{EndoTrans}
\end{equation}
Letting $\bm{z}:=(\bm{z}^{(1)},\bm{z}^{(2)})^{\top}$ where the subvector $\bm{z}^{(k)} := (\eta_{k},\dot{\eta}_{k})^{\top}$ for $k\in\{1,2\}$, the endogenous transformation (\ref{EndoTrans}) allows\footnote{Every differentially flat system can be put in
the Brunovsky normal form \cite[Theorem 4.1]{van1998differential}.} rewriting (\ref{KinemaicBicycleModelSimplified}) in the Brunovsky normal form:
\begin{align}
\dot{\bm{z}} = \underbrace{\blkdiag\left(\bm{A}^{(1)},\bm{A}^{(2)}\right)}_{=:\bm{A}}\bm{z} + \underbrace{\blkdiag\left(\bm{e}_{2}^{2},\bm{e}_{2}^{2}\right)}_{=:\bm{B}}\widetilde{\bm{u}}.
\label{BrunovskyForm4State2Control}	
\end{align}
In (\ref{BrunovskyForm4State2Control}), the matrix $\bm{A}^{(k)}:= \left[\bm{0} \mid \bm{e}_{1}^{2}\right]$ for $k\in\{1,2\}$, and 
\begin{subequations}
\begin{align}
\widetilde{u}_1 &:= a\cos\theta - (v^{2}/\ell)\sin\theta\tan\phi,\\
\widetilde{u}_2 &:= a\sin\theta + (v^{2}/\ell)\cos\theta\tan\phi.
\end{align}
\label{NewControlFromOldControl}	
\end{subequations}
We will need the mapping $\bx \mapsto \bm{z}:= \bm{\tau}(\bm{x}) = (x,v\cos\theta,y,v\sin\theta)^{\top}$, and its inverse $\bm{z} \mapsto \bm{x} = \bm{\tau}^{-1}(\bm{z}) = (z_1, z_3, \atantwo(z_4,z_2), \sqrt{z_2^2 + z_4^2})^{\top}$. In particular, the determinant of the Jacobian
\begin{align}
\det\left(\nabla_{\bx}\bm{\tau}\big\vert_{\bx=\bm{\tau}^{-1}(\bm{z})}\right) = \sqrt{z_{2}^{2} + z_{4}^{2}} \neq 0, \;\text{since}\;v\neq 0.
\label{DetJac}	
\end{align}
This will be useful in the sequel.

%%%%%%%%%%%%%%%%%%%%%%%%%%%%%%%%%%%%%%%%%%%%%%%%%%%%%%%%%%%%%%%%%%%%%%%%%%%%%%%%%%%%%%%%%%%%%%%%%%%%%%%%%%%%%%%%%%%%%

\section{Two Layer Prediction-Control of PDFs}\label{SecFramework}
We next propose a two layer computational framework for the ego vehicle: PDF prediction followed by PDF control. In practice, a third layer at the lower level may exist which
is responsible for safety checks at a higher
frequency than the prediction and control layers.
Based on safety evaluation at this level, corrective actions
can be applied such as lane change abort or emergency braking.

\subsection{PDF Prediction Layer}\label{SubsecPredictionLayer}
The individual vehicle model (\ref{KinemaicBicycleModelSimplified}) is in standard form $\dot{\bm{x}}=\bm{f}(\bm{x},\bm{u})$. We suppose that each car closes the loop with a nominal MPC policy: $\bm{u}=\bm{\pi}_{\rm{MPC}}(\bm{x},t)$, resulting in the closed-loop dynamics $\dot{\bm{x}}=\bm{f}(\bm{x},\bm{\pi}_{\rm{MPC}}(\bm{x},t))$. These MPC policies in each vehicle can be thought of as low-level controllers for holding lanes and maintaining speeds, and compensate the deviations from its mean state at $t=t_{0}$ subject to physical constraints. In Section \ref{SecNumSim}, we will detail an implementation.

At time $t=t_{0}$, the ego vehicle estimates the joint PDFs for its own as well as its neighboring vehicles' states (e.g., using a filtering algorithm) at $t_{0}$. For instance, in the scenario described in Fig. \ref{fig:Schematic}, these joint PDFs will be $\rho_{0}^{\text{ego}}, \rho_{0}^{\text{A}}, \rho_{0}^{\text{L}_{i}}, \rho_{0}^{\text{R}_{j}}$, where $i\in\{1,2,3\}$ and $j\in\{1,2\}$, each of these PDFs being supported on the respective four dimensional $(x,y,\theta,v)$ state spaces. The ego can use these PDFs to forecast the evolution of the respective joint state PDFs over a fixed horizon $[t_{0},T]$ by solving the Liouville PDE for the closed-loop dynamics:
\begin{align}
\!\!\!\frac{\partial\rho}{\partial t} \! + \!\nabla\!\cdot\!\left(\bm{f}(\bm{x},\bm{\pi}_{\rm{MPC}}(\bm{x},t))\rho\right) = 0, \: \rho\left(\bm{x},t=t_0\right)\;\text{known}.
\label{LiouvillePDE}	
\end{align}
The prediction horizon length $T$ is assumed to be small ($\sim 2$ s). For the PDF prediction purpose, the ego assumes that the neighboring non-ego vehicles will not change lane during $[t_{0},t_{0}+T]$.      

In (\ref{LiouvillePDE}), the ego sets the initial joint PDF $\rho\left(\bm{x},t=t_0\right)$ equal to  $\rho_{0}^{\text{ego}}, \rho_{0}^{\text{A}}, \rho_{0}^{\text{L}_{i}}, \rho_{0}^{\text{R}_{j}}$, for the respective vehicles, and predicts their joint PDFs $\rho^{\text{ego}}(\bx,t), \rho^{\text{A}}(\bx,t), \rho^{\text{L}_{i}}(\bx,t), \rho^{\text{R}_{j}}(\bx,t)$ for $t\in[t_0,T]$. Since (\ref{LiouvillePDE}) is a first-order PDE initial value problem, and the characteristic curves of the Liouville PDE are precisely the trajectories of the closed-loop dynamics \cite[Sec. II]{halder2011dispersion}, \cite[Sec. IV.A]{halder2015optimal}, hence the ego can perform the joint PDF predictions via gridless computation as explained next.   

\subsubsection{Propagation of weighted point clouds}\label{subsubsecPropagationPointCloud}
The ego generates $N$ random samples $\{\bm{x}^{i}_{0}\}_{i=1}^{N}$ from each of the known joint PDFs $\rho(\bx,t=t_0)$, and evaluates them at the respective initial joint PDFs to obtain the \emph{weighted} point clouds $\{\bm{x}^{i}_{0},\rho^{i}_{0}\}_{i=1}^{N}$ for each vehicle. It then evolves the weighted $N$-sample point clouds $\{\bm{x}^{i}(t),\rho^{i}(t)\}_{i=1}^{N}$ along the characteristic curves of (\ref{LiouvillePDE}), for each vehicle (including itself). The state vector for the $i$th sample, $\bm{x}^{i}(t)$, is updated via the respective closed-loop dynamics. The joint PDF along the $i$th sample trajectory is updated via the characteristic ODE 
\[\dot{\rho}^{i} = -\nabla_{\bm{x}^{i}}\cdot\bm{f}(\bm{x}^{i},\bm{\pi}_{\rm{MPC}}(\bm{x}^{i},t)), \quad i=1,\hdots,N,\]
i.e., each vehicle's joint PDF propagation requires integrating $5\times 1$ vector ODE with $N$ initial conditions, wherein each of these $N$ samples can be time-propagated in parallel.

\subsubsection{Collision probabilities and Wasserstein barycenters}\label{subsubsecCollisionProbWassBary}
From the joint PDF trajectories $\{\bm{x}^{i}(t),\rho^{i}(t)\}_{i=1}^{N}$, the ego estimates the corresponding bivariate $(x,y)$ marginal PDF trajectories. For the scenario shown in Fig. \ref{fig:Schematic}, these bivariate $(x,y)$ marginal PDFs will be $\rho^{\text{ego}}_{xy}(t), \rho^{\text{A}}_{xy}(t), \rho^{\text{L}_{i}}_{xy}(t), \rho^{\text{R}_{j}}_{xy}(t)$, where $i\in\{1,2,3\}$, $j\in\{1,2\}$. If we denote the space of bivariate $(x,y)$ marginal PDFs as $\mathcal{P}_{xy}$, then we can define collision probability $p_{\text{collision}}$ between two vehicles at any given time $t\in[t_{0},t_{0}+T]$, as a mapping $p_{\text{collision}} : \mathcal{P}_{xy} \times \mathcal{P}_{xy} \mapsto [0,1]$. We refer the readers to \cite[p. 8]{haddad2020density} for the computation of $p_{\text{collision}}$ at any given time.

To decide whether to change the lane or not over the ensuing time horizon $[t_0,t_0+T]$, the ego can compute and compare the $p_{\text{collision}}$ for different feasible scenarios. For the situation depicted in Fig. \ref{fig:Schematic}, if the ego continues in its own lane then the collision probability of interest\footnote{Recall the assumption that non-ego vehicles will not change lanes in $[t_0,t_0+T]$.} is $p_{\text{collision}}(\rho^{\text{ego}}_{xy}(t), \rho^{\text{A}}_{xy}(t))$. If the ego changes lane, either to its left or right, then by the end of the time horizon, it needs to safely situate itself among the available gaps (i.e., longitudinal separations) between the vehicles in these lanes. We suppose that if any of these gaps (in the expected sense) is $\leq 2\ell$ at $t=t_{0}+T$, then the corresponding placement is unsafe. If there are more than one available gaps of expected length $> 2\ell$, then the ego needs to estimate which of them is probabilistically safest for placing itself. A natural way to estimate the same is to compute the Wasserstein barycenter\footnote{The Wasserstein barycenter between a pair of PDFs $\rho_{1},\rho_{2}$ with finite second moments, is defined as the PDF $\rho^{\text{bary}} := \arg\inf_{\rho}\{\lambda_{1}W^{2}(\rho,\rho_1)+\lambda_{2}W^{2}(\rho,\rho_2)\}$, and can be interpreted as the weighted average of the input PDFs $\rho_1,\rho_2$ with weights $\lambda_{i}\geq 0$, $\sum_{i}\lambda_i = 1$. The $\arg\inf$ is taken over all PDFs with finite second moments, and $W(\cdot,\cdot)$ is the Wasserstein metric \cite[Ch. 7]{villani2003topics} on the same space. In our context, $\lambda_1=\lambda_2=0.5$ to ensure equi-Wassersetin separation from $\rho_1,\rho_2$. Notice that if $\rho_{i} = \delta_{(x_i,y_i)}$, $i=1,2$, are Diracs, then $\rho^{\text{bary}} = \delta_{\left(\frac{x_1+x_2}{2},\frac{y_1+y_2}{2}\right)}$.} \cite{agueh2011barycenters} $\rho_{xy}^{\text{bary}}$ of the $(x,y)$ bivariate marginals for the pair of vehicles whose gap is under consideration, and then to compute $p_{\text{collision}}$ between this barycenter and the bivariate marginals for the vehicles in the front and back. Then, the safest placement option for the ego is the gap for which
\par\nobreak
{\small{\begin{align*}
\max\bigg\{p_{\text{collision}}\left(\!\rho_{xy}^{\text{bary}(\text{front},\text{back}\!)},\rho_{xy}^{\text{front}}\!\right), p_{\text{collision}}\!\left(\!\rho_{xy}^{\text{bary}(\text{front},\text{back})},\rho_{xy}^{\text{back}}\!\right)\bigg\}	
\end{align*}}}is minimum among all gaps with expected longitudinal separation $>2\ell$ at $t=t_{0}+T$.

The preceding calculation allows the ego to infer the safest gap in a neighboring lane in the event lane change over $[t_0,t_0+T]$ entails lower collision probability than continuing in its own lane. This inference is done by the ego at $t=t_0$. Let us denote the barycentric \emph{joint} state PDF corresponding to this gap at $t_0+T$ as $\rho_{T}^{\text{desired}}$. The reason for the barycentric joint to be a desired terminal PDF for ego is that it not only allows equi-separation in position coordinates, but also in velocity coordinates, from the vehicles at its front and back. Then it remains to design a feedback controller that will actually transfer the ego from the known initial joint PDF $\rho_0^{\text{ego}}$ to the computed $\rho_{T}^{\text{desired}}$ during $[t_0,t_0+T]$.

\subsection{PDF Control Layer}\label{SubsecControlLayer}
Having obtained the joint PDFs $\rho_{0}^{\text{ego}}$ and $\rho_{T}^{\text{desired}}$ from the prediction layer, the objective of the control layer is to synthesize state feedback $\bm{u}(\bm{x},t)$ that transfers the ego vehicle's \emph{controlled} joint state PDF $\rho^{\bm{u}}(\bm{x},t)$ from $\rho_{0}^{\text{ego}}$ to $\rho_{T}^{\text{desired}}$ over $t\in[t_0,t_0+T]$ subject to the trajectory-level dynamics (\ref{KinemaicBicycleModelSimplified}) while minimizing some ensemble-level stage cost. A natural choice in our context is to minimize
\begin{align}
\mathbb{E}_{\rho^{\bm{u}}}\left[\int_{t_0}^{t_0+T}\frac{1}{2}\left(a^2 + \left(v\dot{\theta}\right)^{2}\right)\:\differential t\right],
\label{OriginalLagrangian}	
\end{align}
i.e., to minimize the average acceleration effort. In (\ref{KinemaicBicycleModelSimplified}), $a$ is the longitudinal acceleration, and $v\dot{\theta}$ is the lateral acceleration under the small sideslip angle assumption.

Using (\ref{NewControlFromOldControl}), the corresponding stochastic optimal control problem in the feedback linearized coordinates becomes
\begin{subequations}
\begin{align}
\underset{\left(\sigma^{\widetilde{\bm{u}}},\widetilde{\bm{u}}\right)}{\inf}\quad &\mathbb{E}_{\sigma^{\widetilde{\bm{u}}}}\left[\int_{t_0}^{t_0+T} \frac{1}{2}\|\widetilde{\bm{u}}\|_{2}^{2}\:\differential t\right]\label{SOCPfeedbacklinearizedObj}\\
\text{subject to}\quad &\frac{\partial\sigma^{\widetilde{\bm{u}}}}{\partial t} + \nabla_{\bm{z}}\cdot\left(\left(\bm{A z} + \bm{B}\widetilde{\bm{u}}\right)\sigma^{\widetilde{\bm{u}}}\right) = 0, \label{SOCPfeedbacklinearizedConstrLiouville}\\
\quad &\sigma^{\widetilde{\bm{u}}}\left(\bm{z},t=0\right) = \bm{\tau}_{\sharp}\rho_{0}^{\text{ego}}, \label{SOCPfeedbacklinearizedConstrInitialPDF}\\
\quad &\sigma^{\widetilde{\bm{u}}}\left(\bm{z},t=T\right) = \bm{\tau}_{\sharp}\rho_{T}^{\text{desired}}.\label{SOCPfeedbacklinearizedConstrTerminalPDF}
\end{align}	
\label{SOCPfeedbacklinearized}	
\end{subequations}
In (\ref{SOCPfeedbacklinearized}), the decision variables are the controlled joint PDF $\sigma^{\widetilde{\bm{u}}}(\bm{z},t)$ and the state feedback $\widetilde{\bm{u}}(\bm{z},t)$. The controlled Liouville PDE (\ref{SOCPfeedbacklinearizedConstrLiouville}) corresponds to the dynamics (\ref{BrunovskyForm4State2Control}). Because $(\bm{A},\bm{B})$ is a controllable pair, problem (\ref{SOCPfeedbacklinearized}) is feasible. Also, since the diffeomorphism $\bm{\tau}$ is known, determining the optimal pair $\left(\sigma^{\widetilde{\bm{u}}}(\bm{z},t),\widetilde{\bm{u}}(\bm{z},t)\right)$ in (\ref{SOCPfeedbacklinearized}) is equivalent to determining the optimal pair $\left(\rho^{\bm{u}}(\bm{x},t),\bm{u}(\bm{x},t)\right)$ in (\ref{OriginalLagrangian}).

In (\ref{SOCPfeedbacklinearizedConstrInitialPDF})-(\ref{SOCPfeedbacklinearizedConstrTerminalPDF}), the pushforward of a given joint PDF $\xi$ via $\bm{\tau}$, i.e., $\zeta := \bm{\tau}_{\sharp}\xi$ can be written explicitly using (\ref{DetJac}) as
\[\zeta(\bm{z}) = \xi\left(z_{1},z_{3},\atantwo(z_4,z_2),\sqrt{z_{2}^{2} + z_{4}^{2}}\right)/\sqrt{z_{2}^{2} + z_{4}^{2}}.\]

Given $t_0,T,\bm{A},\bm{B},\bm{\tau},\rho_{0}^{\text{ego}},\rho_{T}^{\text{desired}}$, problem (\ref{SOCPfeedbacklinearized}) is an instance of the optimal mass transport problem \cite[Sec. 8.1]{villani2003topics} with linear time invariant (LTI) prior dynamics, which admits unique solution provided the pushforwards in the RHS of (\ref{SOCPfeedbacklinearizedConstrInitialPDF})-(\ref{SOCPfeedbacklinearizedConstrTerminalPDF}) have finite second moments \cite[p. 168]{ambrosio2008gradient}. The latter holds in our numerical computation with weighted point clouds having finite supports.

\subsubsection{Stochastic regularization}\label{subsubsecStocRegularization}
To solve (\ref{SOCPfeedbacklinearized}), we fix an $\varepsilon>0$, and replace the RHS of (\ref{SOCPfeedbacklinearizedConstrLiouville}) by $\varepsilon\langle\bm{B}\bm{B}^{\top},{\rm{Hess}}(\sigma^{\widetilde{\bm{u}}})\rangle$, where ${\rm{Hess}}(\sigma^{\widetilde{\bm{u}}})$ denotes the Hessian of $\sigma^{\widetilde{\bm{u}}}$. We can interpret this in two different ways. If the model (\ref{BrunovskyForm4State2Control}) is imperfect in that the noise in actuation was not accounted for, then a natural modification of (\ref{BrunovskyForm4State2Control}) capturing such unmodeled dynamics is
\begin{align}
\differential\bm{z} = \left(\bm{A z} + \bm{B}\widetilde{\bm{u}}\right)\:\differential t + \sqrt{2\varepsilon}\bm{B}\:\differential\bm{w},
\label{ItoSDE}	
\end{align}
where $\bm{w}$ is standard vector Wiener process. In this sense, the parameter $\varepsilon$ models the strength of the actuation noise, and need not be small. On the other hand, if the model (\ref{BrunovskyForm4State2Control}), or equivalently model (\ref{KinemaicBicycleModelSimplified}), is perfect, then one can think of $\varepsilon$ in (\ref{ItoSDE}) as a dynamic stochastic regularization for the computational purpose, and should set a small positive value for it. In either interpretation, the joint PDF dynamics associated with (\ref{ItoSDE}) is precisely
\begin{align}
\frac{\partial\sigma^{\widetilde{\bm{u}}}}{\partial t} + \nabla_{\bm{z}}\cdot\left(\left(\bm{A z} + \bm{B}\widetilde{\bm{u}}\right)\sigma^{\widetilde{\bm{u}}}\right) = \varepsilon\big\langle\bm{B}\bm{B}^{\top},{\rm{Hess}}(\sigma^{\widetilde{\bm{u}}})\big\rangle.
\label{FPKcontrolled} 	
\end{align}
The problem (\ref{SOCPfeedbacklinearized}) with (\ref{SOCPfeedbacklinearizedConstrLiouville}) replaced by (\ref{FPKcontrolled}) amounts to the so-called Schr\"{o}dinger bridge problem subject to LTI prior dynamics which admits unique solution under the foregoing assumption that (\ref{SOCPfeedbacklinearizedConstrInitialPDF})-(\ref{SOCPfeedbacklinearizedConstrTerminalPDF}) have finite second moments. Let us denote the optimizer of (\ref{SOCPfeedbacklinearizedObj}), (\ref{FPKcontrolled}), (\ref{SOCPfeedbacklinearizedConstrInitialPDF}), (\ref{SOCPfeedbacklinearizedConstrTerminalPDF}), as the pair $\left(\sigma^{\widetilde{\bm{u}}}_{\varepsilon},\widetilde{\bm{u}}_{\varepsilon}\right)_{\text{opt}}$ that depends on the choice of the regularization parameter $\varepsilon > 0$. In the limit $\varepsilon \downarrow 0$, it is known \cite{chen2016optimal} that $\left(\sigma^{\widetilde{\bm{u}}}_{\varepsilon},\widetilde{\bm{u}}_{\varepsilon}\right)_{\text{opt}} \rightarrow \left(\sigma^{\widetilde{\bm{u}}},\widetilde{\bm{u}}\right)_{\text{opt}}$, the optimizer of (\ref{SOCPfeedbacklinearized}).

The regularized problem (\ref{SOCPfeedbacklinearizedObj}), (\ref{FPKcontrolled}), (\ref{SOCPfeedbacklinearizedConstrInitialPDF}), (\ref{SOCPfeedbacklinearizedConstrTerminalPDF}) allows \cite[Sec. 4]{chen2016optimal} computing its solution as
\begin{align}
\left(\sigma^{\widetilde{\bm{u}}}_{\varepsilon},\widetilde{\bm{u}}_{\varepsilon}\right)_{\text{opt}} = \left(\widehat{\varphi}(\bm{z},t)\varphi(\bm{z},t),\:2\varepsilon\bm{B}^{\top}\nabla_{\bm{z}}\varphi(\bm{z},t)\right).
\label{OptimalSolutionViaSchrodingerFactors}	
\end{align}
In (\ref{OptimalSolutionViaSchrodingerFactors}), the pair $(\widehat{\varphi},\varphi)$ is given by
\begin{subequations}
\begin{align}
\widehat{\varphi}(\bm{z},t) &=  \int_{\mathcal{Z}_{0}} \kappa(t_{0},\widetilde{\bm{z}},t,\bm{z})\:\widehat{\varphi}_{0}\left(\bm{z}_{0}\right)\differential\bm{z}_{0},\label{phihatattimet}\\
\varphi(\bm{z},t) &= \int_{\mathcal{Z}_{T}}\kappa(t,\bm{z},t_{0}+T,\bm{z}_{T})\:\varphi_{T}\left(\bm{z}_{T}\right)\differential\bm{z}_{T},\label{phiattimet}
\end{align}
\label{TransientFactors}		
\end{subequations}
\noindent where $\widehat{\varphi}_{0}(\cdot):=\widehat{\varphi}(\cdot,t_{0})$, $\varphi_{T}(\cdot):=\varphi(\cdot,t_{0}+T)$, and $\mathcal{Z}_{0},\mathcal{Z}_{T}$ denote the support of the joint PDFs $\bm{\tau}_{\sharp}\rho_{0}^{\text{ego}}$, $\bm{\tau}_{\sharp}\rho_{T}^{\text{desired}}$, respectively. The Markov kernel\footnote{Since $(\bm{A},\bm{B})$ is controllable, $\kappa$ is positive everywhere.} $\kappa$ in (\ref{TransientFactors}) for $t_{0} \leq s < t \leq t_0+T$ is
\begin{align}
\kappa(s,\bm{z},t,\widetilde{\bm{z}}) := &\frac{\det\left(\bm{M}_{ts}\right)^{-1/2}}{(4\pi\varepsilon)^{n/2}}\exp\left(-\frac{1}{4\varepsilon}\times\right. \nonumber\\
&\left.\left(\bm{z}-\bm{\Phi}_{ts}\widetilde{\bm{z}}\right)^{\top}\bm{M}_{ts}^{-1}\left(\bm{z}-\bm{\Phi}_{ts}\widetilde{\bm{z}}\right)\right),
\label{LTIkernel}	
\end{align}
where $\bm{\Phi}_{ts}$, $\bm{M}_{ts}$ are the state transition matrix and the controllability Gramian, respectively, associated with the pair $(\bm{A},\bm{B})$; see (\ref{STM}), (\ref{ControllabGramian}) in Appendix. We next discuss how to solve for the pair $(\widehat{\varphi}_0, \varphi_{T})$, which by (\ref{OptimalSolutionViaSchrodingerFactors}) and (\ref{TransientFactors}), furnishes the solution of the regularized problem (\ref{SOCPfeedbacklinearizedObj}), (\ref{FPKcontrolled}), (\ref{SOCPfeedbacklinearizedConstrInitialPDF}), (\ref{SOCPfeedbacklinearizedConstrTerminalPDF}).

\begin{figure}[t]
\centering
   \includegraphics[width=\linewidth]{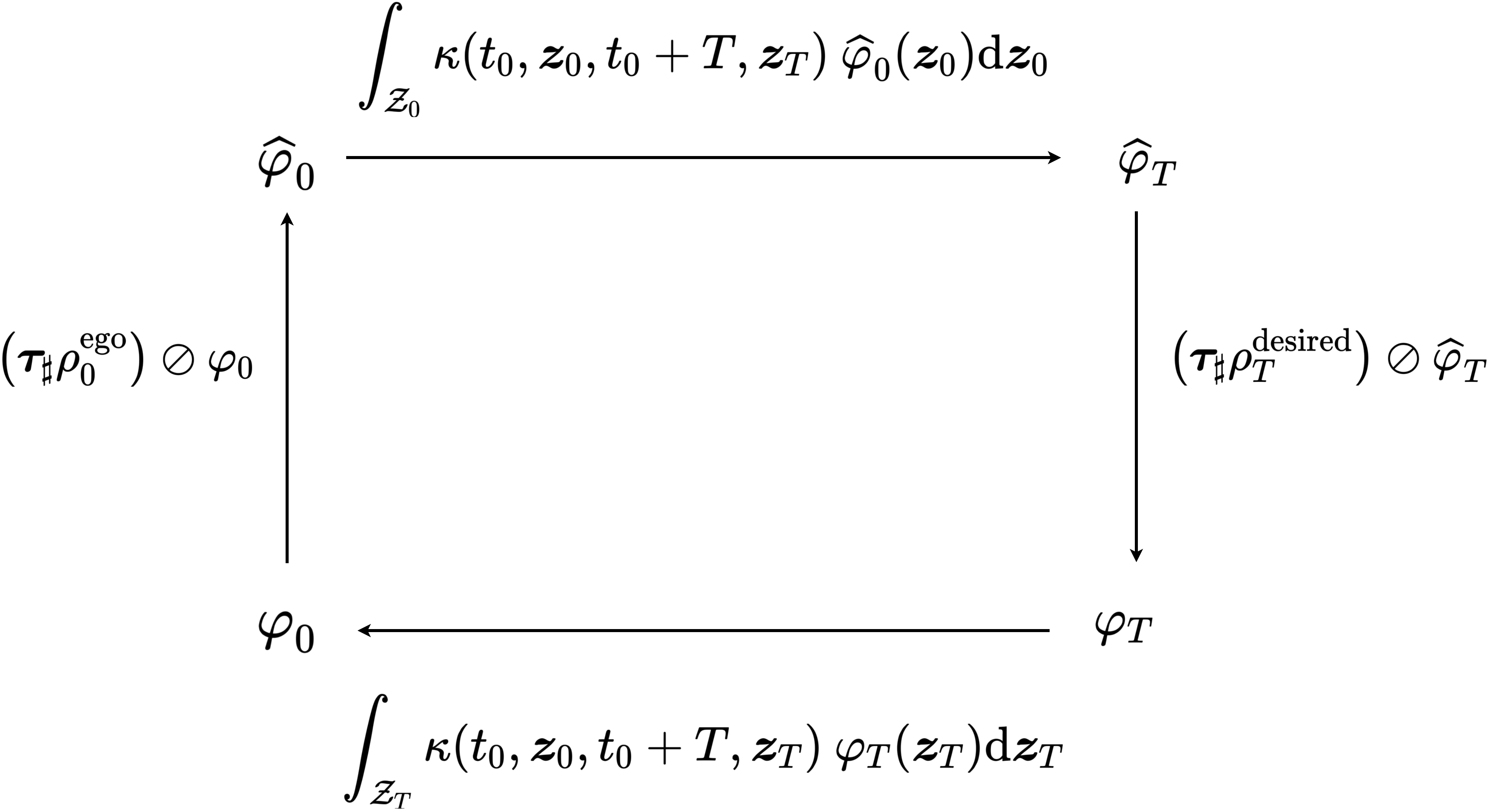}
   \caption{{\small{The fixed point recursion for the pair $\left(\widehat{\varphi}_{0},\varphi_{T}\right)$ using the Markov kernel (\ref{LTIkernel}) and the boundary conditions (\ref{FactorBoundaryconditions}).}}}
\vspace*{-0.1in}   
   \label{fig:FactorRecursion} 
\end{figure}

\subsubsection{Fixed point recursion}\label{subsubsecFixedPointRecursion}
From (\ref{SOCPfeedbacklinearizedConstrInitialPDF})-(\ref{SOCPfeedbacklinearizedConstrTerminalPDF}), the pair $(\widehat{\varphi},\varphi)$ in (\ref{OptimalSolutionViaSchrodingerFactors}) satisfies the boundary conditions
\begin{align}
\widehat{\varphi}_{0}\varphi_{0} = \bm{\tau}_{\sharp}\rho_{0}^{\text{ego}}, \quad \widehat{\varphi}_{T}\varphi_{T} = \bm{\tau}_{\sharp}\rho_{T}^{\text{desired}}.
\label{FactorBoundaryconditions}	
\end{align}
%To compute the pair $(\widehat{\varphi}_0, \varphi_{T})$ in (\ref{TransientFactors}), 
We then set up the fixed point recursion shown in Fig. \ref{fig:FactorRecursion} that is known \cite{chen2016entropic} to be contractive in Hilbert's projective metric, and thus converges to a unique pair $(\widehat{\varphi}_0, \varphi_{T})$, in worst-case linear rate. Having $(\widehat{\varphi}_0, \varphi_{T})$, we use (\ref{TransientFactors}) to compute $(\widehat{\varphi}(\bm{z},t), \varphi(\bm{z},t))$, and then (\ref{OptimalSolutionViaSchrodingerFactors}) to obtain $\left(\sigma^{\widetilde{\bm{u}}}_{\varepsilon},\widetilde{\bm{u}}_{\varepsilon}\right)_{\text{opt}}$. Using the map $\bm{\tau}^{-1}$ in Sec. \ref{SecModels}, we return to $\varepsilon$-regularized version of the original optimizers $(\rho^{\bm{u}},\bm{u})_{\text{opt}}$, denoted as $(\rho^{\bm{u}}_{\varepsilon},\bm{u}_{\varepsilon})_{\text{opt}}$.

We note here that the recursion in Fig. \ref{fig:FactorRecursion} and evaluating (\ref{TransientFactors}) require the Markov kernel (\ref{LTIkernel}), which in turn requires inverse of the \emph{finite horizon} controllability Gramian and its determinant. In general, this amounts to first solving the Lyapunov matrix ODE to compute the finite horizon Gramian itself, and then to compute its inverse and determinant. However, this computational burden can be significantly alleviated by exploiting the structure of the binary matrix pair $(\bm{A},\bm{B})$ in Brunovsky normal form, to analytically compute the inverse and its determinant (Theorem 1 in Appendix \ref{AppThm}).

\begin{figure}[t]
\centering
   \includegraphics[width=\linewidth]{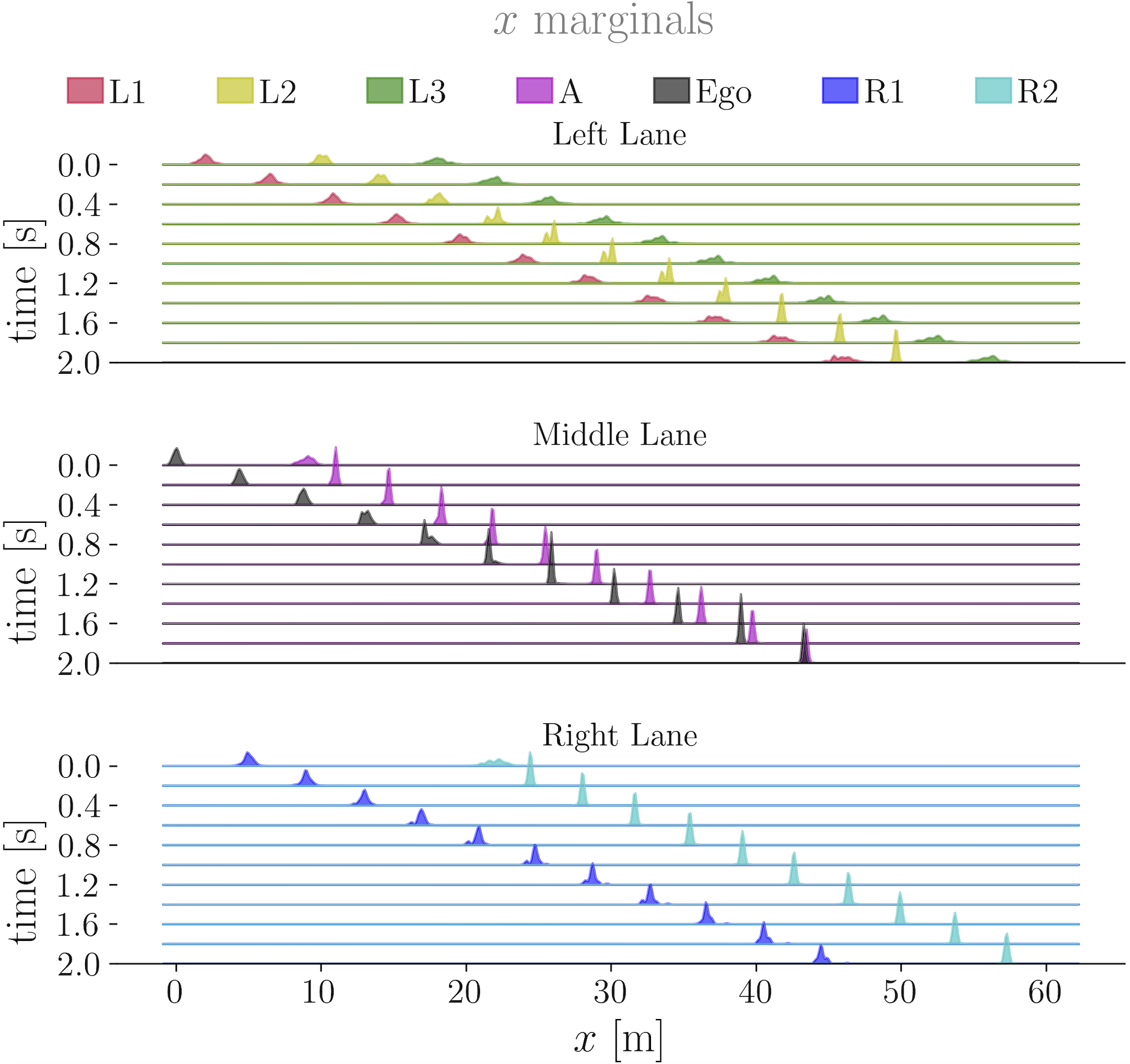}
   \caption{{\small{In a highway driving scenario akin to Fig. \ref{fig:Schematic} (see Sec. \ref{SecNumSim}), the ego vehicle's predictions made at $t=0$ for the $x$-marginal evolutions over time horizon $[0,2]$.}}}
\vspace*{-0.1in}   
   \label{fig:PredictedxMarginals} 
\end{figure}
\begin{figure}[t]
\centering
   \includegraphics[width=0.94\linewidth]{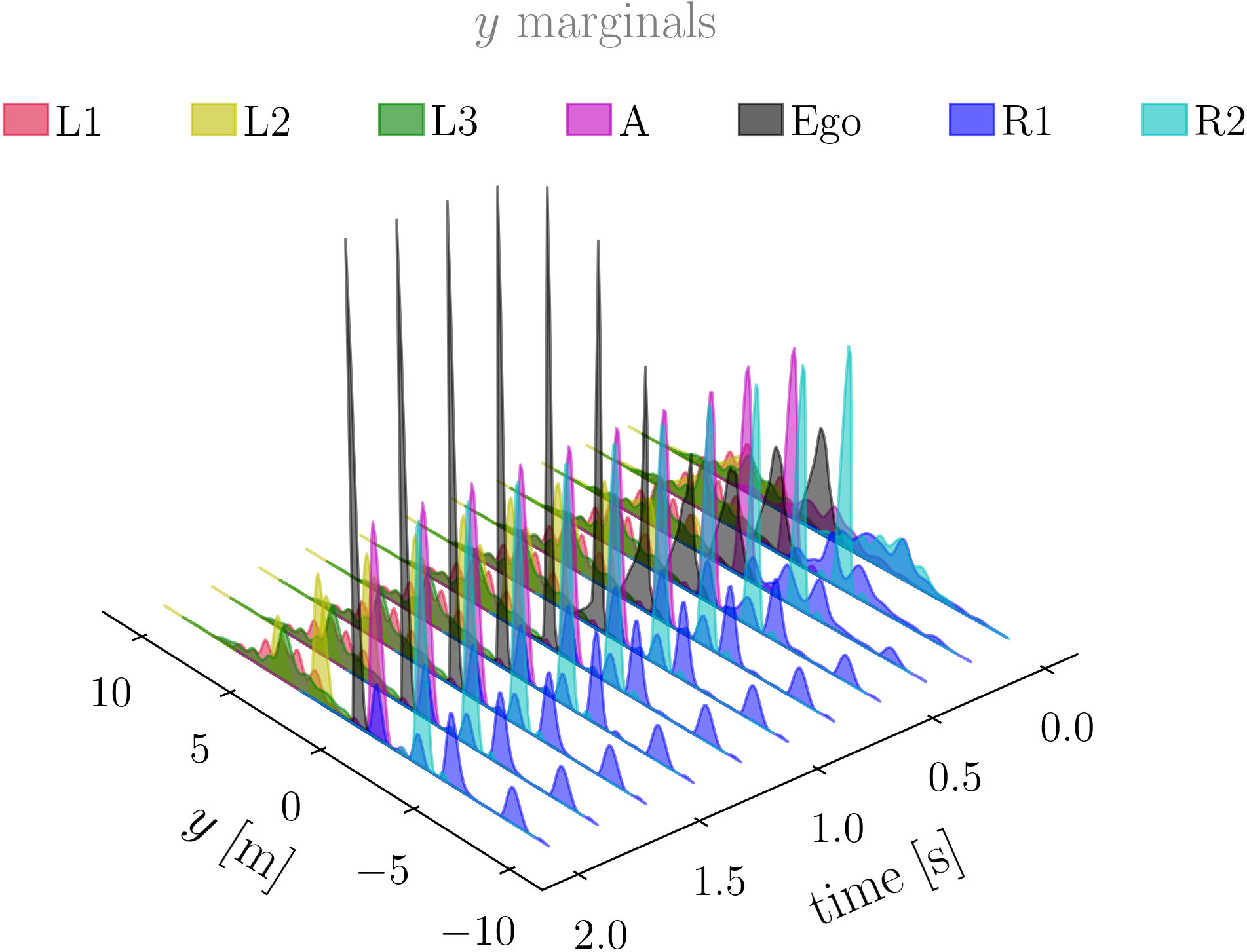}
   \caption{{\small{In a highway driving scenario akin to Fig. \ref{fig:Schematic} (see Sec. \ref{SecNumSim}), the ego vehicle's predictions made at $t=0$ for the $y$-marginal evolutions over time horizon $[0,2]$.}}}
\vspace*{-0.1in}   
   \label{fig:PredictedyMarginals}    
\end{figure}
\begin{figure}[htpb]
\centering
   \includegraphics[width=\linewidth]{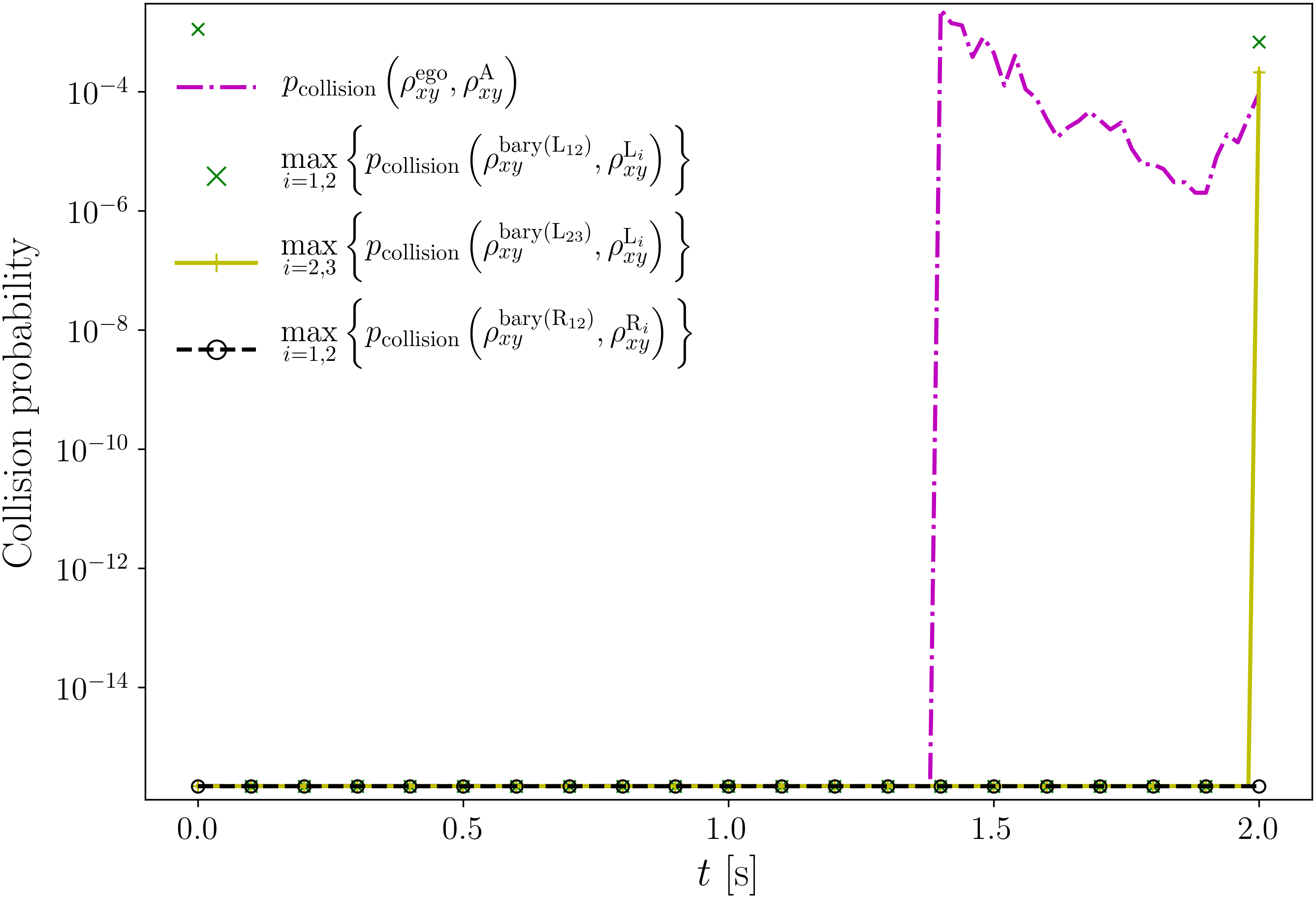}
   \caption{{\small{At $t=0$, the ego vehicle's forecasts of collision probabilities for the highway driving scenario in Sec. \ref{SecNumSim}.}}}
\vspace*{-0.1in}   
   \label{fig:PredictedCollisionProb}    
\end{figure}
\begin{figure}[htpb]
\centering
   \includegraphics[width=\linewidth,height=4cm]{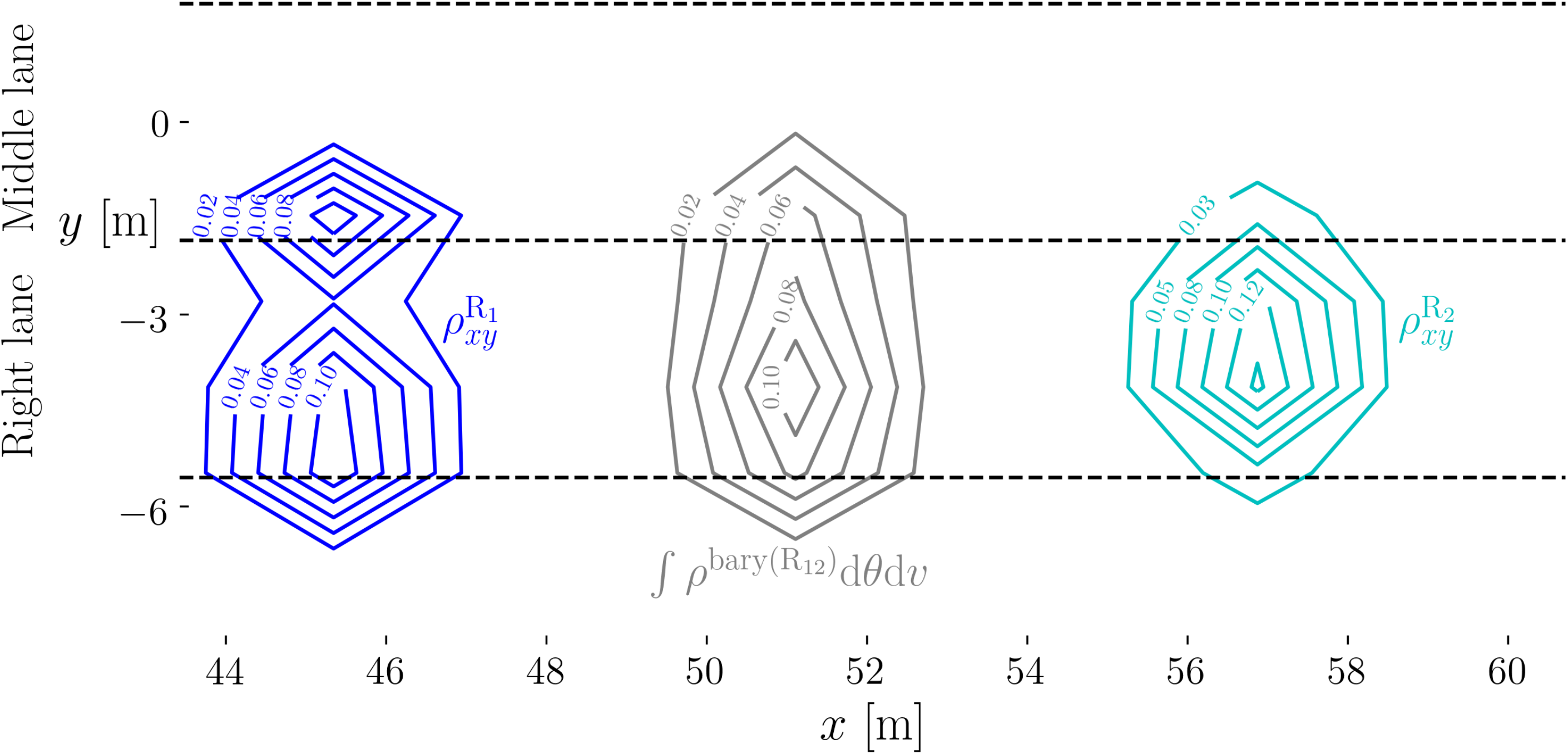}
   \caption{{\small{At $t=2$, the contour plots for the marginals $\rho_{xy}^{\text{R}_{i}}$, $i=1,2$, in the right lane, and the $(x,y)$ marginal of the barycenter of the joints of $R_{1}$ and $R_{2}$, i.e., $\int\rho^{\text{bary(R$_{12}$)}}\differential\theta\differential v$, for the highway driving scenario in Sec. \ref{SecNumSim}. The dashed lines denote the lane boundaries.}}}
\vspace*{-0.1in}   
   \label{fig:RightLaneMarginalOfDesiredBary}    
\end{figure}

%%%%%%%%%%%%%%%%%%%%%%%%%%%%%%%%%%%%%%%%%%%%%%%%%%%%%%%%%%%%%%%%%%%%%%%%%%%%%%%%%%%%%%%%%%%%%%%%%%%%%%%%%%%%%%%%%%%%%

\section{Numerical Simulations}\label{SecNumSim}
To illustrate the two layer framework proposed in Sec. \ref{SecFramework}, we consider a three-lane highway driving scenario as in Fig. \ref{fig:Schematic}, and suppose that the initial joint state PDFs for each of the seven cars: $\rho_{0}^{\text{ego}}, \rho_{0}^{\text{A}}, \rho_{0}^{\text{L}_{i}}, \rho_{0}^{\text{R}_{j}}$, $i\in\{1,2,3\}$, $j\in\{1,2\}$, are jointly Gaussian with parameters detailed in Appendix \ref{AppInitialJointPDF}. We take $N=200$ samples from each and evaluate the respective initial joint PDFs at those samples.

We set the parameter $\ell = 4$ m. The nominal MPC policies mentioned in Sec. \ref{SubsecPredictionLayer} are obtained by linearizing each vehicle's dynamics about the trim\footnote{Here, the trim is defined to be a straight line path in longitudinal direction with constant velocity starting from the respective mean initial conditions.}, the latter computed via \texttt{findop} in Simulink\textsuperscript{\textregistered} subject to the inequality constraints:
\begin{subequations}
\begin{align}
\begin{pmatrix}
-2\:\text{m/s$^2$}\\
-0.5\:\text{degrees}\
\end{pmatrix} \leq &\bm{u} \leq \begin{pmatrix}
2\:\text{m/s$^2$}\\
0.5\:\text{degrees}\
\end{pmatrix}, \label{ControlBounds}\\
\left(\mu_{{y}}\right)_{0} - 1 \leq & {y} \leq \left(\mu_{{y}}\right)_{0} + 1, \label{yBounds}\\
\left(\mu_{v}\right)_{0} - 1 \leq & v \leq \left(\mu_{v}\right)_{0} + 1. \label{vBounds}
\end{align}
\label{TrimConstraints}
\end{subequations}
The constraint (\ref{ControlBounds}) enforces bounded controls. The path constraint (\ref{yBounds}) enforces the lateral position $y$ to be within $\pm 1$ m of the initial mean lateral position $\left(\mu_{{y}}\right)_{0}$, while (\ref{vBounds}) enforces the velocity $v$ to be within $\pm 1$ m/s of the initial mean velocity $\left(\mu_{v}\right)_{0}$. Notice that the computation of the trim depends on the initial joint state PDFs.

For each vehicle, we linearize (\ref{KinemaicBicycleModelSimplified}) about the respective trim point $(\bm{x}_{\text{trim}},\bm{u}_{\text{trim}})$ to obtain the corresponding LTI matrix pair $\left(\bm{A}_{\text{trim}},\bm{B}_{\text{trim}}\right)$, which is used to compute the explicit MPC feedback $\bm{u} = \bm{\pi}_{\text{MPC}}(\bm{x},t)$ for this linearized system via the MPC toolbox \cite{bemporad2010model} minimizing the deviation from $(\bm{x}_{\text{trim}},\bm{u}_{\text{trim}})$ over a prediction horizon of length $t_p = 2$ s subject to (\ref{TrimConstraints}). The MPC objective is set to minimize a quadratic cost of the form \cite[eq. (11)]{haddad2020density} with the state, control and slew rate weights $\bm{Q} = 10\bm{I}_{4}, \bm{R}= \bm{I}_{2}, \bm{S} = 10^{-1}\bm{I}_{2}$, respectively, and with a sampling time $0.02$ s. This \emph{offline} nominal explicit MPC synthesis results in continuous piecewise affine feedback $\bm{\pi}_{\text{MPC}}(\bm{x},t) = \bm{\Gamma}_{q}\bm{x} + \bm{\gamma}_{q}$ for $\bm{x}(t)\in\mathcal{R}_{q}$, where $q=1,2,\hdots,\nu$, and $\sqcup_{q=1}^{\nu}\mathcal{R}_{q}$ is a disjoint polytopic partition of the of the reach set of the respective closed-loop constrained LTI system. The gain matrix-vector pairs $\{\bm{\Gamma}_{q},\bm{\gamma}_{q}\}_{q=1}^{\nu}$ and the polytopic partitions are stored for each vehicle. In our simulation with the seven vehicles, the number of partitions $\nu$ were 282--301.

\begin{figure}[t]
\centering
\includegraphics[width=\linewidth,height=3.5cm]{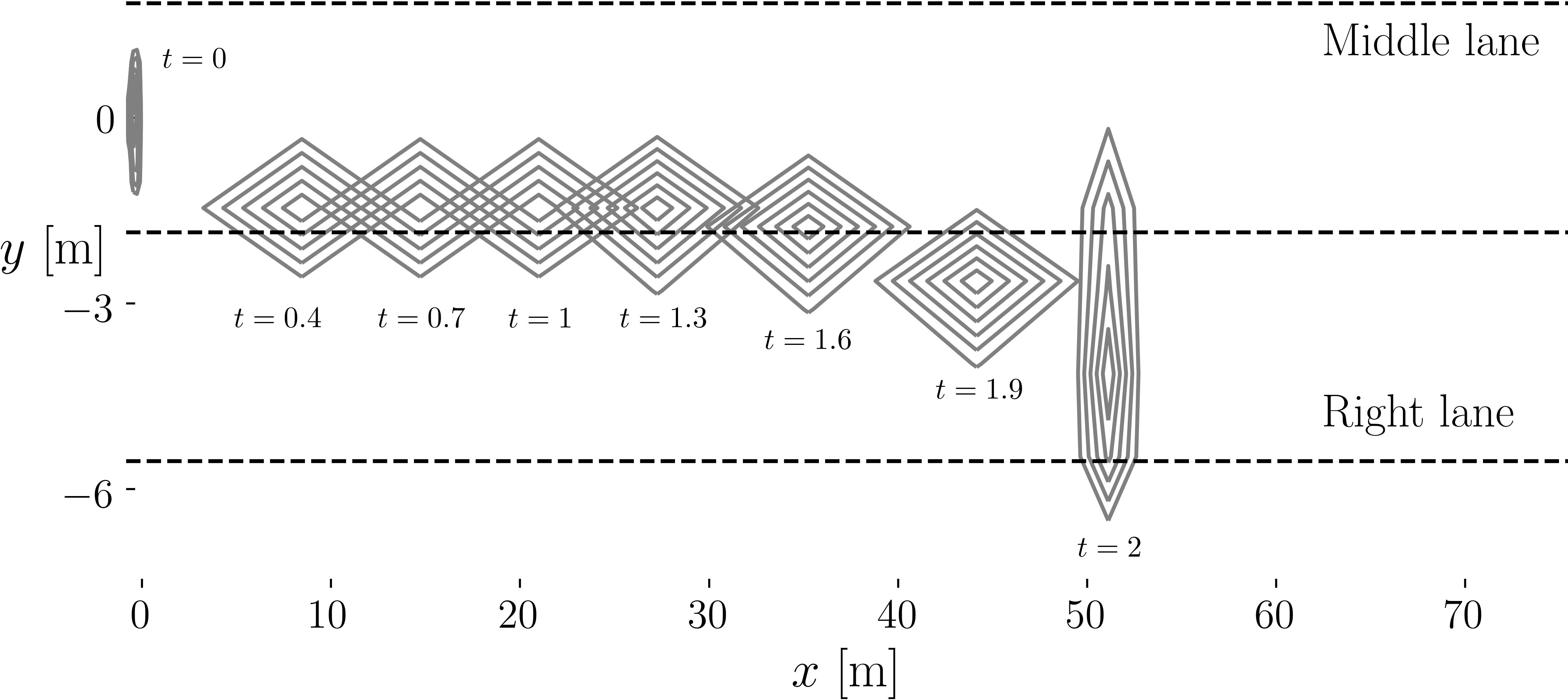}
 \caption{{\small{Bivariate $(x,y)$ marginals of the ego vehicle's optimally controlled joint PDFs from $\rho_{0}^{\text{ego}}(\bm{x})$ to $\rho_{T}^{\text{desired}}(\bm{x})$ over $t\in[0,2]$.}}}
\vspace*{-0.2in}   
   \label{fig:controlledxyMarginals}    
\end{figure}

Starting with the respective initial joint PDFs mentioned before, and using the framework in Sec. \ref{subsubsecPropagationPointCloud} with $N=200$ samples for each, the ego predicts the corresponding transient joint PDFs $\rho^{\text{ego}}(\bx,t), \rho^{\text{A}}(\bx,t), \rho^{\text{L}_{i}}(\bx,t), \rho^{\text{R}_{j}}(\bx,t)$ from $t_{0}=0$ to $t_{0}+T=2$ s, subject to (\ref{KinemaicBicycleModelSimplified}) with the respective nominal MPC policies in the loop. Fig. \ref{fig:PredictedxMarginals} and \ref{fig:PredictedyMarginals} show the corresponding univariate marginals in $x$ and $y$ positions, respectively. Notice from the middle lane plot in Fig. \ref{fig:PredictedxMarginals} that the longitudinal separation between the ego and the vehicle A decreases over time (due to larger initial mean velocity of the ego compared to A), indicating that continuation in the same lane, for the next 2 s, will be unsafe for the ego. The collision probability forecasts made by the ego at $t=0$, as described in Sec. \ref{subsubsecCollisionProbWassBary}, are shown in Fig. \ref{fig:PredictedCollisionProb}, which corroborates that for the ego to continue in the middle lane is unsafe. Fig. \ref{fig:PredictedCollisionProb} also predicts that the safest option for the ego is to initiate a lane change to its right at $t=0$, so as to situate itself between R\textsubscript{1} and R\textsubscript{2} by $t=2$, specifically at the barycenter of $\rho^{\text{R}_{1}}(\bx,t=2)$ and $\rho^{\text{R}_{2}}(\bx,t=2)$, denoted as $\rho_{2}^{\text{desired}}(\bm{x})$. The $(x,y)$ bivariate marginal of this desired terminal barycentric joint PDF is shown in Fig. \ref{fig:RightLaneMarginalOfDesiredBary}. For the computation of the Wasserstein barycentric PDF, we used the multi-marginal Sinkhorn algorithm from \cite[Sec. 4.2]{benamou2015iterative}.

The transfer of the ego vehicle's joint PDF, from $\rho_{0}^{\text{ego}}(\bm{x})$ to $\rho_{2}^{\text{desired}}(\bm{x})$ over $t\in[0,2]$, is then accomplished in real time via $\widetilde{\bm{u}}_{\varepsilon}(\bm{z},t)$ synthesized in feedback linearized coordinates, as detailed in Sec. \ref{subsubsecStocRegularization}. Recall from Sec. \ref{subsubsecFixedPointRecursion} that this feedback synthesis, in turn, reduces to solving a fixed point recursion outlined in Fig. \ref{fig:FactorRecursion}, which uses Theorem 1 in Appendix for computing the Markov kernel (\ref{LTIkernel}). Fig. \ref{fig:controlledxyMarginals} shows the optimally controlled $(x,y)$ bivariate marginals $\int(\rho^{\bm{u}}_{\varepsilon})_{\text{opt}}\differential\theta \differential v$ illustrating the finite horizon density regulation. Some details on the computation related to the feedback synthesis is given in Appendix \ref{AppSBPdetails}.   

%%%%%%%%%%%%%%%%%%%%%%%%%%%%%%%%%%%%%%%%%%%%%%%%%%%%%%%%%%%%%%%%%%%%%%%%%%%%%%%%%%%%%%%%%%%%%%%%%%%%%%%%%%%%%%%%%%%%%

\section{Concluding Remarks}\label{SecConclusions}
In this paper, we have proposed a computational framework to predict and control the transient joint PDFs of the stochastic states of the vehicles for safe automated driving. Direct nonparametric prediction and control of the joint PDFs allow stochastic robustness in real-time decision making. Although we focused on multi-lane highway driving scenarios for specificity, the proposed framework can be adapted to other scenarios of practical interest, such as decision making in signalized traffic intersections. These extensions will comprise our future work.

%%%%%%%%%%%%%%%%%%%%%%%%%%%%%%%%%%%%%%%%%%%%%%%%%%%%%%%%%%%%%%%%%%%%%%%%%%%%%%%%%%%%%%%%%%%%%%%%%%%%%%%%%%%%%%%%%%%%%%

\appendix

\subsection{On the Finite Horizon Controllability Gramian for the Brunovsky Normal Form}\label{AppThm}
The purpose of this Appendix is to derive some results on the finite horizon controllability Gramian for the Brunovsky normal form, to help speed up the computation in Sec. \ref{SubsecControlLayer}.

Suppose a nonlinear control system with $n$ states and $m$ controls has vector relative degree $\bm{\pi}=(\pi_{1}, \pi_{2}, ..., \pi_{m})^{\top}$ with $\pi_{1}+\pi_2 + ... + \pi_{m} = n$, and that it can be put in the Brunovsky normal form
\begin{align}
\!\dot{\bm{z}} = \underbrace{\blkdiag\!\left(\bm{A}^{(1)}, ..., \bm{A}^{(m)}\right)}_{=:\bm{A}}\!\bm{z} + \underbrace{\blkdiag\!\left(\bm{e}_{\pi_{1}}^{\pi_{1}}, ...,\bm{e}_{\pi_{m}}^{\pi_{m}}\right)}_{=:\bm{B}}\!\widetilde{\bm{u}},	
\label{GeneralBrunovsky}
\end{align}
where $\bm{z} := \left(\bm{z}^{(1)}, \hdots, \bm{z}^{(m)}\right)^{\!\top}\!\in\mathbb{R}^{n}$, the subvector $\bm{z}^{(k)}\in\mathbb{R}^{\pi_{k}}$ for $k\in\{1,...,m\}$. The state matrix $\bm{A}$ in (\ref{GeneralBrunovsky}) comprises of the diagonal blocks 
\[\bm{A}^{(k)} := \left[\bm{0} \mid \bm{e}_{1}^{\pi_{k}} \mid \bm{e}_{2}^{\pi_{k}} \mid ... \mid \bm{e}_{\pi_{k}-1}^{\pi_{k}}\right] \quad\text{for}\quad k\in\{1,...,m\}.\]
Notice that (\ref{BrunovskyForm4State2Control}) is an instance of (\ref{GeneralBrunovsky}) with $n=4$, $m=2$, $\pi_1 = \pi_2 = 2$.

Let $t_{0}\leq s < t \leq t_{0}+T$. The state transition matrix for (\ref{GeneralBrunovsky}) in time horizon $[s,t]$ is
\begin{align}
&\!\!\bm{\Phi}_{ts} :=\bm{\Phi}(t,s) = \exp(\bm{A}(t-s)) \nonumber\\
&\!\!= \blkdiag\left(\exp(\bm{A}^{(1)}(t-s)), ..., \exp(\bm{A}^{(m)}(t-s))\right),
\label{STM}	
\end{align}
wherein for each $k\in\{1,...,m\}$, we have \cite[Appendix A]{haddad2020convex}
\begin{align}
\exp(\bm{A}^{(k)}(t-s)) = \begin{cases}
 \frac{(t-s)^{j-i}}{(j-i)!} & \text{for}\quad i<j,\\
 1 & \text{for}\quad i=j,\\
 0 & \text{for}\quad i>j,	
 \end{cases}
\label{STMblock} 	
\end{align}
for $i,j = 1, ..., \pi_{k}$. 

The controllability Gramian in time horizon $[s,t]$ is
\begin{align}
\bm{M}_{ts} := \bm{M}(t,s) = \int_{s}^{t}\bm{\Phi}_{t\tau}\bm{B}\bm{B}^{\top}\bm{\Phi}_{t\tau}^{\top}\:\differential\tau.
\label{ControllabGramian}	
\end{align}
Since the pair $(\bm{A},\bm{B})$ in (\ref{GeneralBrunovsky}) is controllable, $\bm{M}_{ts}$ is strictly positive definite. Changing the integration variable in (\ref{ControllabGramian}) from $\tau$ to $t-\tau$, a direct computation using (\ref{STM})-(\ref{STMblock}) yields 
\begin{align}
\bm{M}_{ts} = \blkdiag\left(\bm{M}_{ts}^{(1)}, ..., \bm{M}_{ts}^{(m)}\right),
\label{ControllabGramianBlockDiag}	
\end{align}
whose $k$th diagonal block, for $k\in\{1,...,m\}$, is given by
\begin{align}
\left[\bm{M}_{ts}^{(k)}\right]_{ij} = \frac{(t-s)^{2\pi_{k} - i - j + 1}}{(\pi_{k} - i)!(\pi_{k} - j)!(2\pi_{k}-i-j+1)}
\label{ControllabGramianBlockElementwise}	
\end{align}
for $i,j=1,...,\pi_{k}$. 

Looking at (\ref{ControllabGramianBlockDiag}) and (\ref{ControllabGramianBlockElementwise}), it is far from obvious how to get analytical handle on the determinant and inverse of $\bm{M}_{ts}$, which can speed up the evaluation of (\ref{LTIkernel}) in Fig. \ref{fig:FactorRecursion} and in (\ref{TransientFactors}). Theorem \ref{ThmDetInverseOfControllabGramian} gives explicit formula for the same.

\begin{theorem}\label{ThmDetInverseOfControllabGramian}
(\textbf{Explicit formula for the determinant and inverse of $\textbf{M}_{ts}$}) Let $\bm{M}_{ts}$ be the controllability Gramian (\ref{ControllabGramian}) for the LTI system (\ref{GeneralBrunovsky}) over time horizon $[s,t]$. Then \\
(i) $\det\left(\bm{M}_{ts}\right) = \prod_{k=1}^{m} \big\{ (t-s)^{\pi_{k}^2} \prod_{r=1}^{\pi_{k}} \frac{\Gamma(r)}{\Gamma(\pi_{k}+r)} \big\}$.\\
(ii) $\bm{M}_{ts}^{-1} = \!\blkdiag\left(\left(\bm{M}_{ts}^{(1)}\right)^{-1}, ..., \left(\bm{M}_{ts}^{(m)}\right)^{-1}\right)$, where for $k\in\{1,...,m\}$ and $i,j = 1,...,\pi_{k}$,
\begin{align*}
\!\left[\left(\bm{M}_{ts}^{(k)}\right)^{-1}\right]_{ij} &= \frac{(\pi_{k}-i)! (\pi_{k}-j)!}{(2\pi_{k} - i - j + 1)(t-s)^{2\pi_{k}-i-j+1}}\:\times\\
&\frac{\prod_{r=1}^{\pi_{k}}(2\pi_{k} - i - r + 1)(2\pi_{k} - j - r + 1)}{\left(\displaystyle\prod_{\stackrel{r=1}{r\neq i}}^{\pi_{k}}(r-i)\right) \left(\displaystyle\prod_{\stackrel{r=1}{r\neq j}}^{\pi_{k}}(r-j)\right)}.	
\end{align*}
\end{theorem}
\begin{proof}
(i) Without loss of generality, let $\pi_{k} \geq 2$ since for $\pi_{k}=1$, the matrix $\bm{M}_{ts}^{(k)}$ is a scalar (equal to its trivial determinant). We rewrite (\ref{ControllabGramianBlockElementwise}) as
\begin{align}
\left[\bm{M}_{ts}^{(k)}\right]_{ij} = \int_{0}^{t-s} \frac{\sigma^{\pi_{k}-i}}{(\pi_{k}-i)!}\frac{\sigma^{\pi_{k}-j}}{(\pi_{k}-j)!}	\differential\sigma.
\label{IntegralRepresentation}
\end{align}
Recall Andr\'{e}ief identity\footnote{This can be seen as the continuum version of the
Cauchy-Binet formula.} \cite{andreief1883note,forrester2019meet}, cf. \cite[part II, problem 68]{polya1998problems}: for two sequences of integrable functions $\{f_{i}(\sigma)\}_{i=1}^{\nu}$ and $\{g_{i}(\sigma)\}_{i=1}^{\nu}$, we have
{\small{\begin{align*}
&\int\!\!\hdots\!\int\!\!\det\left(f_{i}(\sigma_{j})\right)\det\left(g_{i}(\sigma_{j})\right)\differential \sigma_{1}\hdots\differential \sigma_{\nu} \\
&= \nu!\:\det\left(\int f_{i}(\sigma)g_{j}(\sigma)\differential\sigma\right),	
\end{align*}
which applied to (\ref{IntegralRepresentation}), yields
\begin{align}
&\det\!\left(\!\bm{M}_{ts}^{(k)}\!\right) \!=\! \frac{1}{\pi_{k}!}\int_{0}^{t-s}\!\!\!\!\!\hdots\int_{0}^{t-s}\!\!\left(\!\!\det\!\left(\frac{\sigma_{j}^{\pi_{k}-i}}{(\pi_{k}-i)!}\right)\!\!\right)^{\!\!2}\!\!\differential\sigma_1\hdots\differential\sigma_{\pi_{k}}\nonumber\\
&= \frac{1}{\pi_{k}!}\!\int_{0}^{t-s}\!\!\!\!\!\!\!\hdots\!\!\int_{0}^{t-s}\!\!\left(\!\frac{1}{\prod_{i=1}^{\pi_{k}}(\pi_{k}-i)!}\!\!\prod_{1 \leq i < j \leq \pi_{k}}\!\!\!\!(\sigma_{j}-\sigma_{i})\!\right)^{\!\!2}\!\!\differential\sigma_1\hdots\differential\sigma_{\pi_{k}}\nonumber\\
&= \underbrace{\frac{1}{\pi_{k}!\left(\prod_{i=1}^{\pi_{k}}(\pi_{k}-i)!\right)^{2}}}_{=:p_k}\underbrace{\int_{0}^{t-s}\!\!\!\!\!\!\!\hdots\!\!\int_{0}^{t-s}\!\!\!\!\!\!\prod_{1 \leq i < j \leq \pi_{k}}\!\!\!\!\!(\sigma_{j}-\sigma_{i})^{2}\differential\sigma_1\hdots\differential\sigma_{\pi_{k}}}_{=:I_k}.
\label{DetOfOneBlock}	
\end{align}}}
In the above, the first step used that for (\ref{IntegralRepresentation}), $f_{i}(\cdot)=g_{i}(\cdot) = (\cdot)^{\pi_{k}-i}/(\pi_{k}-i)!$. The second step utilized a Vandermonde-type determinant. The third step moved the square inside the product in integrand.

Using the change of variable $\varsigma_{r} := \sigma_{r}/(t-s)$ for $r=1,...,\pi_{k}$, the integral $I_k$ in (\ref{DetOfOneBlock}) reduces to a special case of the Selberg integral \cite{selberg1944berkninger,forrester2008importance}: 
\begin{align*}
I_k &= (t-s)^{\pi_{k} + \pi_{k}(\pi_{k}-1)}\!\!\int_{0}^{1}\!\!\!\!\!\hdots\!\!\int_{0}^{1}\!\!\prod_{1 \leq i < j \leq \pi_{k}}\!\!\!\!(\varsigma_{j}-\varsigma_{i})^{2}\:\differential\varsigma_1\hdots\differential\varsigma_{\pi_{k}}\\
&=  (t-s)^{\pi_{k}^2}\prod_{i=0}^{\pi_{k}-1}\frac{(\Gamma(i+1))^{2}\:\Gamma(i+2)}{\Gamma(\pi_{k}+i+1)}.	
\end{align*}
Expressing the pre-factor $p_k$ in (\ref{DetOfOneBlock}) in terms of the Gamma functions, we get 
\[p_k = \frac{1}{\Gamma(\pi_{k}+1)\prod_{r=1}^{\pi_{k}}(\Gamma(r))^2}.\]
Therefore, (\ref{DetOfOneBlock}) gives 
\[\det\!\left(\!\bm{M}_{ts}^{(k)}\!\right) = p_k I_k = (t-s)^{\pi_{k}^2}\prod_{r=1}^{\pi_{k}}\frac{\Gamma(r)}{\Gamma(\pi_{k}+r)}.\] 
From (\ref{ControllabGramianBlockDiag}), we also have $\det\left(\bm{M}_{ts}\right) =  \prod_{k=1}^{m}\det\left(\bm{M}_{ts}^{(k)}\right)$. Combining the last two statements, we obtain the result.

(ii) For $k\in\{1,...,m\}$, define $\bm{\alpha}^{(k)}\in\mathbb{R}^{\pi_{k}}$ with components $[\bm{\alpha}^{(k)}]_{i} := (t-s)^{\pi_{k}-i}/(\pi_{k}-i)!$. Notice from (\ref{ControllabGramianBlockElementwise}) that
\begin{align}
\bm{M}_{ts}^{(k)} = \diag(\bm{\alpha}^{(k)})\widetilde{\bm{M}}^{(k)} \diag(\bm{\alpha}^{(k)}),
\label{DiagScaling}	
\end{align}
where $\left[\widetilde{\bm{M}}^{(k)}\right]_{ij} := \frac{t-s}{2\pi_{k}-i-j+1}$ for all $i,j=1,...,\pi_{k}$. Hence 
\begin{align}
\left[\!\left(\bm{M}_{ts}^{(k)}\right)^{\!-1}\!\right]_{ij} = \frac{(\pi_{k}-i)!(\pi_{k}-j)!}{(t-s)^{2\pi_{k}-i-j}}\left[\!\left(\widetilde{\bm{M}}^{(k)}\right)^{\!-1}\!\right]_{ij}.
\label{InverseGeneralForm}	
\end{align}
That $\widetilde{\bm{M}}^{(k)}$ is nonsingular follows from (\ref{DiagScaling}) as $\bm{M}_{ts}^{(k)}$ is positive definite\footnote{This in turn follows from the positive definiteness of $\bm{M}_{ts}$ and (\ref{ControllabGramianBlockDiag}): a block diagonal matrix is positive definite iff its diagonal blocks are positive definite.}, and $\bm{\alpha}^{(k)}$ is element-wise positive. 

To determine the inverse of $\widetilde{\bm{M}}^{(k)}$, write it as a scaled Cauchy matrix: $\left[\widetilde{\bm{M}}^{(k)}\right]_{ij} = (t-s)/(a_i + b_j)$, where $a_{i} := \pi_{k}-i$, and $b_{j}:=\pi_{k}-j+1$. This enables computation of the inverse \cite[Sec. 1.2.3, Exercise 41]{knuth1997fundamental},\cite{schechter1959inversion}:
\begin{align*}
\left[\!\left(\widetilde{\bm{M}}^{(k)}\right)^{\!-1}\!\right]_{ij} &= \frac{1}{(t-s)(2\pi_{k} - i - j + 1)}\:\times\\
&\frac{\prod_{r=1}^{\pi_{k}}(2\pi_{k} - i - r + 1)(2\pi_{k} - j - r + 1)}{\left(\displaystyle\prod_{\stackrel{r=1}{r\neq i}}^{\pi_{k}}(r-i)\right) \left(\displaystyle\prod_{\stackrel{r=1}{r\neq j}}^{\pi_{k}}(r-j)\right)},	
\end{align*}
which together with (\ref{InverseGeneralForm}), completes the proof.
\end{proof}

To illustrate Theorem \ref{ThmDetInverseOfControllabGramian}, consider $\bm{\pi} = (3, 2)^{\top}$, $s=0$, $t=1$. Using the Theorem, we then have
\begin{align*}
\bm{M}_{10}^{-1} = \begin{pmatrix}
 720 & -360 & 60 & 0 & 0\\
 -360 & 192 & -36 & 0 & 0\\
 60 & -36 & 9 & 0 & 0\\
 0 & 0 & 0 & 12 & -6\\
 0 & 0 & 0 & -6 & 4	
 \end{pmatrix},	
\end{align*}
and $\det\left(\bm{M}_{10}^{-1}\right) = 103680$.

To get a sense of the numerical benefit offered by Theorem \ref{ThmDetInverseOfControllabGramian}, we compared the computational times for constructing the LTI Markov kernel (\ref{LTIkernel}) in two different ways. The first way was to use the Lyapunov matrix ODE via \texttt{ode45} in MATLAB to compute the finite horizon Gramian $\bm{M}_{ts}$, and then using that $\bm{M}_{ts}$ for constructing (\ref{LTIkernel}). The second way was to directly use the formulae of the inverse and determinant of $\bm{M}_{ts}$ from Theorem \ref{ThmDetInverseOfControllabGramian}, in constructing (\ref{LTIkernel}). For $s=0$, $t=1$, $\varepsilon=0.5$, the results are summarized in Table I wherein the state dimension for the first row is $n=4$, and the same for the second row is $n=5$. Both rows correspond to $m=2$ controls. In each case, the domain was $[-1,1]^{n}$ with 5 uniform discretization per dimension, i.e., $5^{n}$ samples. All simulations were done via MATLAB R2019b on iMac with 3.4 GHz Quad-Core Intel Core i5 processor and 8 GB memory.
\begin{table}[ht]\label{tab:computationaltimes}
  \centering
  \renewcommand{\arraystretch}{1.5}
\begin{tabular}{| c | c | c |}
   \hline
   \multirow{2}{3cm}{ Vector relative degree $\bm{\pi}$ } & \multicolumn{2}{ c |}{Computational time [s]} \\    \cline{2-3} & \;using Lyapunov ODE\; & \;using Theorem \ref{ThmDetInverseOfControllabGramian}\;\\ \hline\hline
 $(2,2)^{\top}$ & 1.9556 & 0.2995\\ \hline
  $(3,2)^{\top}$ & 49.7869 & 6.9294\\ \hline 
   \end{tabular}
     \caption{Computational times in evaluating (\ref{LTIkernel}).}
\vspace*{-0.15in}   	
   	\end{table}

%%%%%%%%%%%%%%%%%%%%%%%%%%%%%%%%%%%%%%%%%%%%%%%%%%%%%%%%%%%%%%%%%%%%%%%%%%%%%%%%%%%%%%%%%%%%%%%%%%%%

\subsection{Initial Joint PDFs in Section \ref{SecNumSim}}\label{AppInitialJointPDF}
We suppose that the initial joint state PDFs for each of the seven cars: $\rho_{0}^{\text{ego}}, \rho_{0}^{\text{A}}, \rho_{0}^{\text{L}_{i}}, \rho_{0}^{\text{R}_{j}}$, $i\in\{1,2,3\}$, $j\in\{1,2\}$, are jointly Gaussian with respective initial mean vectors:
\par\nobreak
\vspace*{-0.15in}
{\small{\begin{align*}
&\bm{\mu}_{0}^{\text{L}_1} = \left(2, 3.7, 0, 22\right)^{\!\top}\!\!, \bm{\mu}_{0}^{\text{L}_2} = \left(10, 3.7, 0, 20\right)^{\!\top}\!\!, \bm{\mu}_{0}^{\text{L}_3} = \left(18, 3.7, 0, 19\right)^{\!\top}\!\!, \\
&\bm{\mu}_{0}^{\text{ego}} = \left(0, 0, 0, 22\right)^{\top}, \quad \bm{\mu}_{0}^{\text{A}} = \left(9, 0, 0, 18\right)^{\top}, \\
&\bm{\mu}_{0}^{\text{R}_1} = \left(5, -3.7, 0, 20\right)^{\top}, \quad \bm{\mu}_{0}^{\text{R}_2} = \left(22, -3.7, 0, 18\right)^{\top},
\end{align*}}}
and respective initial covariance matrices:
\begin{align*}
\bm{\Sigma}_{0}^{\text{L}_1} &= {\rm{diag}}\left(0.44, 4, 2.7\times10^{-6}, 0.16\right),\\
\bm{\Sigma}_{0}^{\text{L}_2} &= {\rm{diag}}\left(0.25, 7.1, 2.7\times10^{-6}, 0.11\right),\\
\bm{\Sigma}_{0}^{\text{L}_3} &= {\rm{diag}}\left(1, 7.1, 2.7\times10^{-6}, 0.16\right),\\
\bm{\Sigma}_{0}^{\text{ego}} &= {\rm{diag}}\left(0.11 ,0.44, 2.7\times10^{-6}, 0.03\right),\\
\bm{\Sigma}_{0}^{\text{A}} &= {\rm{diag}}\left(0.44, 7.1, 2.7\times10^{-6}, 0.13\right),\\
\bm{\Sigma}_{0}^{\text{R}_1} &= {\rm{diag}}\left(0.25, 7.1, 2.7\times10^{-6}, 0.11\right),\\
\bm{\Sigma}_{0}^{\text{R}_2} &= {\rm{diag}}\left(1, 5.4, 2.7\times10^{-6}, 0.11\right).
\end{align*}

%%%%%%%%%%%%%%%%%%%%%%%%%%%%%%%%%%%%%%%%%%%%%%%%%%%%%%%%%%%%%%%%%%%%%%%%%%%%%%%%%%%%%%%%%%%%%%%%%%%%

\subsection{Simulation Details for Controller Synthesis in Section \ref{SecNumSim}}\label{AppSBPdetails}
We next give some details of the parameters used for the feedback synthesis steering the initial joint PDF $\rho_{0}^{\text{ego}}(\bm{x})$ to $\rho_{2}^{\text{desired}}(\bm{x})$ over $t\in[0,2]$. In our context, the desired terminal PDF $\rho_{2}^{\text{desired}}(\bm{x})$ was computed as the Wasserstein barycenter between the joint PDFs of R\textsubscript{1} and R\textsubscript{2} at $t=2$ with $\lambda_{1}=\lambda_2 = 0.5$. 

We used the stochastic regularization $\varepsilon = 0.1$. To implement the fixed point recursion mentioned in Section \ref{subsubsecFixedPointRecursion}, we used the closed form formulae derived in Theorem \ref{ThmDetInverseOfControllabGramian} in (\ref{LTIkernel}). To assess the convergence of the fixed point recursion for the pair $(\widehat{\varphi}_{0},\varphi_{T})$, we employed Hilbert's projective metric $d_{\text{Hilbert}}$ (see e.g., \cite{lemmens2014birkhoff}) between the previous and the current iterates. Recall that for a given pair of element-wise positive vectors $\bm{p},\bm{q}$ (i.e., $p_{i},q_{i}>0$ for all $i$), we have
\[d_{\text{Hilbert}}\left(\bm{p},\bm{q}\right) = \log\left(\dfrac{\underset{i}{\max}(\bm{p}\oslash\bm{q})}{\underset{i}{\min}(\bm{p}\oslash\bm{q})}\right).\]
\noindent While the number of iterations in the fixed point recursion was less than 1000, we checked if $d_{\text{Hilbert}}(\widehat{\varphi}_{0}^{\text{previous}},\widehat{\varphi}_{0}^{\text{current}})$ and $d_{\text{Hilbert}}(\varphi_{T}^{\text{previous}},\varphi_{T}^{\text{current}})$ were both less than a given numerical tolerance: $10^{-4}$. In our simulation, this fixed point recursion converged in 73 iterations. 

To avoid the loss of floating point precision, we performed the fixed point recursion shown in Fig. \ref{fig:FactorRecursion} in logarithmic domain using the standard log-sum-exp trick.

\bibliography{references.bib}
\bibliographystyle{IEEEtran}

\end{document}